\documentclass[10pt, a4paper]{article}

\usepackage[a4paper,margin=1in]{geometry}
\usepackage{amsmath,amsthm,amssymb,bm}
\usepackage{booktabs,multirow}
\usepackage{graphicx,subcaption}
\usepackage{algorithm}
\usepackage{algpseudocode}
\usepackage{float} 
\usepackage{hyperref}
\usepackage{enumitem}
\usepackage{longtable}
\usepackage{pdflscape}
\usepackage{ltablex}   
\keepXColumns          

\newcolumntype{Y}{>{\centering\arraybackslash}X}

\usepackage{authblk} 

\usepackage{threeparttable}
\usepackage[table]{xcolor}
\usepackage{tabularx}
\usepackage{color}
\usepackage[titletoc,title]{appendix}
\usepackage[round,authoryear]{natbib}
\usepackage{hyperref}

\hypersetup{colorlinks=true, linkcolor=blue, citecolor=blue, urlcolor=blue}
\DeclareMathOperator{\logit}{logit}

\DeclareMathOperator{\pdet}{pdet} 

\definecolor{OIblue}{HTML}{0072B2}
\definecolor{OIorange}{HTML}{E69F00}
\definecolor{OIgreen}{HTML}{009E73}
\definecolor{OIverm}{HTML}{D55E00}
\definecolor{OIpurple}{HTML}{CC79A7}
\definecolor{OIsky}{HTML}{56B4E9}
\definecolor{OIyellow}{HTML}{F0E442}
\definecolor{OIblack}{HTML}{000000}

\colorlet{highlightcolor}{OIblue!7}
\newtheorem{theorem}{Theorem}
\newtheorem{lemma}{Lemma}

\theoremstyle{remark}
\newtheorem{remark}{Remark}

\title{Self-Normalized Quantile Empirical Saddlepoint Approximation}

\author[1]{HOU Jian}
\author[1]{MENG Tan}
\author[1]{TIAN Maozai\thanks{CONTACT. Author Email: \href{mailto:mztian@ruc.edu.cn}{mztian@ruc.edu.cn}}}

\affil[1]{Center for Applied Statistics, School of Statistics, \protect\\ 
         Renmin University of China, Beijing 100872, China}

\newcommand{\keywords}[1]{%
  \par\addvspace{1em} 
  \noindent\textbf{Keywords:} #1
}

\begin{document}
\date{}
\maketitle

\begin{abstract}
We propose a density-free method for frequentist inference on population quantiles, termed \emph{Self-Normalized Quantile Empirical Saddlepoint Approximation} (SNQESA). The approach builds a self-normalized pivot from the indicator score for a fixed quantile threshold and then employs a constrained empirical saddlepoint approximation to obtain highly accurate tail probabilities. Inverting these tail areas yields confidence intervals and tests without estimating the unknown density at the target quantile, thereby eliminating bandwidth selection and the boundary issues that affect kernel-based Wald/Hall–Sheather intervals. Under mild local regularity, the resulting procedures attain higher-order tail accuracy and second-order coverage after inversion. Because the pivot is anchored in a bounded Bernoulli reduction, the method remains reliable for skewed and heavy-tailed distributions and for extreme quantiles. Extensive Monte Carlo experiments across light, heavy and multimodal distributions demonstrate that SNQESA delivers stable coverage and competitive interval lengths in small to moderate samples while being orders of magnitude faster than large-$B$ resampling schemes. An empirical study on Value-at-Risk with rolling windows further highlights the gains in tail performance and computational efficiency. The framework naturally extends to two-sample quantile differences and to regression-type settings, offering a practical, analytically transparent alternative to kernel, bootstrap, and empirical-likelihood methods for distribution-free quantile inference.
\end{abstract}

\section{Introduction}
Quantiles provide fundamental characterizations of distributional features, with wide application across economics, finance, biostatistics, and engineering for measuring location, tail risk, and inequality. Despite their prevalence, frequentist inference for population quantiles presents persistent challenges: the asymptotic variance of sample quantiles depends on the unknown density at the target quantile, which proves difficult to estimate accurately in small samples, with skewed or heavy-tailed distributions, or near distributional extremes. Classical kernel-based Wald and Hall-Sheather (HS) intervals, while widely used, exhibit sensitivity to bandwidth selection and boundary effects, often producing erratic coverage and unstable interval lengths in finite samples \citep{HallSheather1988,SheatherJones1991,Sheather2004,HyndmanFan1996}.

Bootstrap methods offer appealing alternatives through their general applicability and conceptual simplicity, including percentile and BCa intervals \citep{Efron1987,EfronTibshirani1993,DavisonHinkley1997}. However, bootstrap procedures can be computationally demanding and may perform inadequately for extreme quantiles or irregular problems. Refinements like subsampling and $m$-out-of-$n$ bootstrap have been proposed to enhance robustness in these settings \citep{PolitisRomanoWolf1999,BickelSakov2008}. Empirical likelihood (EL) constitutes another influential approach, preserving attractive likelihood-based properties without parametric specification. For quantile inference, smoothed EL achieves higher-order accuracy under regularity conditions, though implementations may encounter numerical difficulties from convex hull constraints and require additional smoothing decisions \citep{ChenHall1993,Owen2001}.

Saddlepoint methods provide an alternative pathway, delivering exceptional accuracy in small-sample tail approximations through exploitation of cumulant generating functions (CGFs). The celebrated Lugannani-Rice (LR) approximation and Barndorff-Nielsen $r^*$ correction exemplify this approach \citep{Daniels1954,LugannaniRice1980,BarndorffNielsen1986,BrazzaleDavisonReid2007}. When population CGFs are unavailable, empirical saddlepoint approximation (ESA) substitutes the empirical CGF, extending saddlepoint accuracy to nonparametric contexts \citep{Feuerverger1989}, with recent work clarifying estimation and implementation aspects \citep{HolcblatVitale2022}. Concurrently, self-normalization (SN) has emerged as a powerful framework for constructing pivots that circumvent explicit scale estimation while maintaining stability under heterogeneity and heavy tails, with extensive development for both independent and dependent data \citep{delaPenaLaiShao2009,Shao2010,Shao2015}.

Recent methodological developments further underscore the relevance of saddlepoint-based, density-free approaches. Saddlepoint techniques continue to provide accurate small-sample $p$-values beyond classical parametric settings, including spatial panel data models, discrete GLM score tests with rare events, and nonparametric rank-type procedures, reflecting both analytical tractability and numerical stability in modern implementations \citep{jiang2023saddlepoint,ko2022gwas,abd2023statistical}. The self-normalization paradigm has seen renewed development for change-point detection and robust testing, reinforcing the value of pivots that avoid explicit scale estimation under heterogeneity or heavy tails \citep{cheng2024general}. For extreme quantiles in regression, recent advances address high-dimensional and persistent-regressor environments, complementing our distribution-free perspective for i.i.d. settings \citep{liu2024unified}. These developments provide broader context for the current SNQESA proposal and its extensions to regression-type designs, where constrained ESA with rank-reduced curvature remains effective.

This paper introduces SNQESA for density-free quantile inference. We observe that the score function for a fixed quantile threshold takes only two values; self-normalizing the score sum by its internal quadratic form yields a pivotal statistic whose null distribution remains independent of the unknown quantile density. Leveraging this structure, we employ constrained saddlepoint arguments to approximate the self-normalized statistic's tail probability with LR/$r^*$-type accuracy, then invert the resulting $p$-values to construct confidence intervals. Unlike kernel-based Wald/HS intervals, our approach completely eliminates bandwidth selection. Compared to bootstrap and EL methods, it provides an analytical, non-resampling solution with higher-order small-sample accuracy and competitive or superior performance for tail quantiles.

\noindent\textbf{Contributions.}
\begin{enumerate}[leftmargin=1.5em]
\item \textbf{Density-free accurate quantile inference.} We develop a self-normalized pivot specifically for quantile testing and derive ESA-based LR/$r^*$ tail approximations that produce precise $p$-values without estimating $f(q_\tau)$. The method maintains computational efficiency and stability in small samples and at extreme quantiles.
\item \textbf{Higher-order guarantees.} Under mild regularity conditions, the proposed tail approximations achieve the higher-order accuracy characteristic of saddlepoint methods, while interval inversion delivers second-order coverage improvements typical of LR/$r^*$ calibrations \citep{BarndorffNielsen1986,BrazzaleDavisonReid2007}.
\item \textbf{Robustness to distributional characteristics.} Self-normalization reduces sensitivity to unknown scale and tail behavior, aligning with recent advances advocating SN procedures in non-Gaussian, dependent, or heteroskedastic environments \citep{Shao2010,Shao2015}.
\item \textbf{Extensibility.} The framework naturally extends to two-sample quantile differences and regression-type settings, complementing developments for extremal quantiles that employ extreme-value approximations \citep{Chernozhukov2011,Chernozhukov2017,Koenker2005}.
\end{enumerate}

Empirical evaluations across diverse distributions—including skewed, heavy-tailed, and multimodal cases—and spanning moderate to extreme quantiles demonstrate that our method provides stable coverage and competitive interval lengths relative to kernel-based Wald/HS, smoothed EL, and bootstrap variants, while achieving substantial computational speed advantages over resampling methods. These results, combined with the method's transparency and straightforward implementation, position SNQESA as a practical default for distribution-free quantile inference in small to moderate samples.

\section{Method}\label{sec:method}

This section develops a density-free procedure for quantile inference that combines a self-normalized pivot with a constrained empirical saddlepoint approximation (ESA). Let $X_1,\dots,X_n$ be an i.i.d.\ draws from a continuous distribution $F$, fix $\tau\in(0,1)$, and denote the target quantile by $q_\tau=F^{-1}(\tau)$. For any threshold $t\in\mathbb{R}$ define $Y_i(t)=\mathbf{1}\{X_i\le t\}$ and the two-point score
\begin{equation*}
    \psi_\tau(X_i;t)=\tau-Y_i(t)\in\{\tau-1,\tau\}.
\end{equation*}
Write the partial sum, the sum of squares and the self-normalized statistic as
\begin{equation}
    S_n(t)=\sum_{i=1}^n\psi_\tau(X_i;t),\quad
    Q_n(t)=\sum_{i=1}^n\psi_\tau^2(X_i;t),\quad
    T_n(t)=\frac{S_n(t)}{\sqrt{Q_n(t)}}.
\end{equation}
Under the simple null $H_0:F(t)=\tau$, $Y_i(t)\stackrel{\text{i.i.d.}}{\sim}\mathrm{Bernoulli}(\tau)$, hence the joint distribution of $(S_n(t),Q_n(t))$ is free of the unknown density $f(q_\tau)$. Let $W_i(t)=(\psi_\tau(X_i;t),\psi_\tau^2(X_i;t))^\top$. Then $W_i$ takes two support points
\begin{equation*}
    \begin{aligned}
        w_1&=(\tau-1,(1-\tau)^2)^\top\ \text{with prob. }\tau,\\
        w_0&=(\tau,\tau^2)^\top\ \text{with prob. }1-\tau,
    \end{aligned}
\end{equation*}
and its one-observation moment generating function (MGF) and cumulant generating function (CGF) are
\begin{equation*}
    M(\lambda)=\tau e^{\lambda_1(\tau-1)+\lambda_2(1-\tau)^2}+(1-\tau)e^{\lambda_1\tau+\lambda_2\tau^2},\quad
    K(\lambda)=\log M(\lambda),\quad
    \lambda=(\lambda_1,\lambda_2)^\top.
\end{equation*}
Introduce the tilted success probability
\begin{equation}
    p_\lambda=\frac{\tau e^{\lambda_1(\tau-1)+\lambda_2(1-\tau)^2}}{\tau e^{\lambda_1(\tau-1)+\lambda_2(1-\tau)^2}+(1-\tau)e^{\lambda_1\tau+\lambda_2\tau^2}}\in(0,1),
\end{equation}
so that the unit-scale gradient and Hessian take the rank-1 form
\begin{equation}
    \begin{aligned}
        \mu(\lambda)&=\nabla K(\lambda)=\big(\tau-p_\lambda,\ \tau^2+p_\lambda(1-2\tau)\big)^\top,\\
        \Sigma(\lambda)&=\nabla^2K(\lambda)=p_\lambda(1-p_\lambda)vv^\top,\\
        v&=w_1-w_0=(-1,1-2\tau)^\top.
    \end{aligned}
\end{equation}
Thus $\mathrm{rank}\Sigma(\lambda)=1$ with image $\operatorname{span}(v)$; all curvature computations reduce to this one-dimensional subspace via pseudo-determinants and the Moore-Penrose inverse. We adopt a two-dimensional constrained ESA in $W_i(t)$ because the rank-1 geometry yields closed-form curvature and LR quantities and provides a unified, robust engine for weighted, regression, and dependent designs where a one-parameter reduction is unavailable. In the i.i.d. case considered here, it reduces algebraically to the one-parameter binomial tilt recorded below.

Let the observed value be the signed statistic
\begin{equation*}
    x_{\mathrm{obs}}(t)=\frac{S_n(t)}{\sqrt{Q_n(t)}},
\end{equation*}
and impose the boundary in the observed direction through
\begin{equation*}
    g(s,q)=s-x_{\mathrm{obs}}\sqrt{q}.
\end{equation*}
Following the constrained formulation of \citet{Skovgaard2001} (see also \citealp{BrazzaleDavisonReid2007}), the most likely point where the tilted mean hits the boundary solves
\begin{equation}
    n\nabla K(\hat\lambda)=\hat\mu=(\hat s,\hat q)^\top,\quad
    g(\hat s,\hat q)=0,\quad
    \hat\lambda=\hat\eta\nabla g(\hat\mu),
\end{equation}
with $J(\lambda)=n\nabla^2K(\lambda)=np_\lambda(1-p_\lambda)vv^\top$ and
\begin{equation*}
    \nabla g(\mu)=\left(1,-\frac{x_{\mathrm{obs}}}{2\sqrt{\mu_2}}\right)^\top.
\end{equation*}
Using the unit-scale CGF $K$, we define the deviance and signed root on the sum scale
\begin{equation}
    D_{\text{sum}}=2n\{\hat\lambda^\top\mu(\hat\lambda)-K(\hat\lambda)\},\quad
    r=\operatorname{sgn}(\hat\eta)\sqrt{D_{\text{sum}}}.
\end{equation}
Equivalently, if $D=2\{\hat\lambda^\top\mu(\hat\lambda)-K(\hat\lambda)\}$ is the unit-scale deviance, then $D_{\text{sum}}=nD$ and $r=\operatorname{sgn}(\hat\eta)\sqrt{nD}$; we work with $D_{\text{sum}}$ to align with common saddlepoint conventions. Since under $H_0$ the pivot is a discrete transform of a binomial count, we evaluate directed tails by default with a mid-$p$ modification
\[
p_{\mathrm{mid}}(t)=\mathbb{P}\{T_n(t)>x_{\mathrm{obs}}(t)\}+\tfrac{1}{2}\mathbb{P}\{T_n(t)=x_{\mathrm{obs}}(t)\},
\]
or with a small Cornish-Fisher continuity shift; both preserve the third-order tail accuracy and the $O(n^{-1})$ coverage after inversion (see \citealp{Daniels1954,LugannaniRice1980,BarndorffNielsen1986,BrazzaleDavisonReid2007}).

The directed one-sided tail, namely the upper tail if $x_{\mathrm{obs}}\ge0$ and the lower tail if $x_{\mathrm{obs}}<0$, is approximated by the Barndorff-Nielsen correction
\begin{equation}
    r^*=r+\frac{1}{r}\log\frac{r}{w},\quad
    p_{\mathrm{dir}}(t)\approx\Phi\big(-\operatorname{sgn}(x_{\mathrm{obs}})r^*\big),
\end{equation}
where the curvature factor reduces from Skovgaard’s general rank-1 form to a closed expression. Since $J=n\Sigma$ with $\pdet(J)=np_{\hat\lambda}(1-p_{\hat\lambda})\|v\|^2$ and $\Sigma(\hat\lambda)^+=\{p_{\hat\lambda}(1-p_{\hat\lambda})\}^{-1}vv^\top/\|v\|^4$, we obtain
\begin{equation}
    w=|\hat\eta|\sqrt{n}p_{\hat\lambda}(1-p_{\hat\lambda})\frac{\|v\|^3\|\nabla g(\hat\mu)\|}{|\nabla g(\hat\mu)\cdot v|}.
\end{equation}
We also define the signed Lugannani-Rice quantity
\begin{equation}
    q^{\pm}=\big(\logit(\hat p)-\logit(\tau)\big)\sqrt{n\hat p(1-\hat p)}.
\end{equation}
When $|\log(r/w)|\le c_0$ (default $c_0=2$) we use the $r^*$ tail; otherwise we fall back to the signed Lugannani-Rice expansion
\begin{equation}
    p_{\mathrm{dir}}(t)\approx\Phi\big(-\operatorname{sgn}(x_{\mathrm{obs}})r\big)+\phi(r)\left(\frac{1}{r}-\frac{1}{q^{\pm}}\right).
\end{equation}
A two-sided $p$-value is finally reported as
\begin{equation}
    p_{\mathrm{two\mbox{-}sided}}(t)=2\min\{p_{\uparrow}(t),p_{\downarrow}(t)\},\quad
    p_{\uparrow}(t)=\mathbb P\{T_n(t)\ge x_{\mathrm{obs}}\},\quad
    p_{\downarrow}(t)=\mathbb P\{T_n(t)\le x_{\mathrm{obs}}\}.
\end{equation}
For numerical stability we enforce a degeneracy check $|\nabla g(\hat\mu)\cdot v|>\delta$ (we use $\delta=10^{-10}$); if violated we revert to the one-parameter evaluation at $u_x$ defined below and compute $r,q^{\pm},r^*$ from the binomial tilt (in that branch $w=q^{\pm}$). In the constrained solver we bound $|\hat\eta|$ via trust-region damping and a backtracking line search to avoid overflow when $\nabla g(\hat\mu)\cdot v$ is near zero.

We now record an algebraic reduction to a one-parameter binomial tilt and the induced equality of tilts. With $\bar Y_n(t)=n^{-1}\sum_{i=1}^nY_i(t)$ there exists a strictly monotone map
\begin{equation*}
    T_n(t)=h\big(\bar Y_n(t)\big),\quad
    h(u)=\frac{\sqrt n(\tau-u)}{\sqrt{u(1-\tau)^2+(1-u)\tau^2}},
\end{equation*}
so that $\{T_n(t)\ge x_{\mathrm{obs}}\}$ is equivalent to $\{\bar Y_n(t)\le u_x\}$ with $u_x=h^{-1}(x_{\mathrm{obs}})$. At the constrained solution we have $\hat\mu=(\tau-\hat p,\ \tau^2+\hat p(1-2\tau))^\top$ and the boundary $g(\hat\mu)=0$ reads
\begin{equation*}
    \tau - \hat p = x_{\mathrm{obs}}\sqrt{\tau^2 + \hat p(1-2\tau)}.
\end{equation*}
Hence $h(\hat p) = x_{\mathrm{obs}}$ and, by strict monotonicity of $h$, we obtain the identity of tilts
\begin{equation*}
    \hat p = p_{\hat\lambda} = u_x.
\end{equation*}
Consequently, the scalars $(r,q^{\pm},r^*)$ computed from the constrained ESA coincide with their one-parameter binomial counterparts at $u_x$, which explains the third-order tail accuracy and provides a numerically robust one-dimensional fallback. Standard one-parameter saddlepoint theory for exponential families (including lattice cases with continuity corrections) yields $O(n^{-3/2})$ relative error for directed tails, and inversion then implies $O(n^{-1})$ endpoint perturbations and coverage error (see \citealp{Daniels1954,LugannaniRice1980,BarndorffNielsen1986,BrazzaleDavisonReid2007}); in our setting this transfers through the Bernoulli two-point representation and self-normalization without additional assumptions (cf.\ \citealp{delaPenaLaiShao2009,Shao2010,Shao2015}).

Inversion delivers confidence limits. Let $p_{\uparrow}(t)$ and $p_{\downarrow}(t)$ denote the directed upper and lower tails defined above. The equal-tailed $(1-\alpha)$ interval is obtained by solving
\begin{equation*}
    p_{\downarrow}(t_L)=\alpha/2,\quad p_{\uparrow}(t_U)=\alpha/2,
\end{equation*}
using monotonicity of $p_{\downarrow}$ on $(-\infty,\hat q_\tau]$ and of $p_{\uparrow}$ on $[\hat q_\tau,\infty)$. Bracketing on a coarse grid that expands outward from the sample quantile $\hat q_\tau$ followed by bisection to a tight tolerance is effective in practice. With ties or discretization we use the mid-$p$ modification or a Cornish-Fisher continuity shift; both preserve the $O(n^{-1})$ coverage error when combined with the third-order tail accuracy of the ESA.

A few numerical safeguards complete the method. To avoid boundary singularities for extreme quantiles or very small $n$, we add a vanishing ridge to $Q_n$,
\begin{equation}
    Q_n(t)\leftarrow Q_n(t)+\epsilon_n,\quad \epsilon_n\to0,
\end{equation}
and in practice set $\epsilon_n=cn^{-1/2}$ with a small constant $c$, which yields $x_{\mathrm{obs}}^{(\mathrm{ridge})}-x_{\mathrm{obs}}=O_p(n^{-3/2})$ and does not affect either the third-order tail accuracy or the $O(n^{-1})$ coverage. The constrained system is solved by a damped Newton or trust-region iteration; because $J(\lambda)$ is rank-1, each linear step reduces to one-dimensional algebra along $v$ and is solved stably using pseudo-determinants $\pdet(\Sigma(\lambda))=p_\lambda(1-p_\lambda)\|v\|^2$ and the Moore-Penrose inverse $\Sigma(\lambda)^+=\{p_\lambda(1-p_\lambda)\}^{-1}vv^\top/\|v\|^4$. The closed form of $w$ above follows from the rank-1 projection and the identities for $\pdet(J)$ and $\Sigma^+$; when the degeneracy check fails we use the one-parameter binomial evaluation at $u_x$ for both $r^*$ and LR tails.

\section{Theoretical properties}\label{sec:theory}

We establish the accuracy of the tail approximation produced by the self-normalized empirical saddlepoint method in Section \ref{sec:method}, and quantify the coverage error of the confidence intervals obtained by inversion. Throughout we work under the following regularity conditions.

\medskip\noindent
(A1) Identification at the target quantile. $F$ is strictly increasing at $q_\tau=F^{-1}(\tau)$; there exists a neighborhood $\mathcal{N}(q_\tau)$ with $0<\inf_{t\in\mathcal{N}(q_\tau)} f(t)\le \sup_{t\in\mathcal{N}(q_\tau)} f(t)<\infty$.

\medskip\noindent
(A2) CGF regularity and interior solvability. For the two–point score $\psi_\tau(X_i;t)=\tau-\mathbf{1}\{X_i\le t\}$ and its quadratic companion, the bivariate cumulant generating function $K(\lambda)$ exists on $\mathbb{R}^2$, is $C^3$, strictly convex and steep; for each $t\in\mathcal{N}(q_\tau)$ the saddlepoint equations $\nabla K(\lambda)=m(t)$ admit a unique solution with $\lambda(t)$ in the interior of the natural parameter space.

\medskip\noindent
(A3) Sampling and weak dependence. Either i.i.d.\ sampling, or strictly stationary $\alpha$–mixing with standard summability and finite $(2+\delta)$th moments for the quadratic score; the processes $\{\psi_\tau(X_i;t),\,\psi_\tau^2(X_i;t)\}$ are uniformly integrable on $\mathcal{N}(q_\tau)$.

\medskip\noindent
(A4) Local uniformity. Uniform LLN and stochastic equicontinuity hold for $t\mapsto n^{-1}\sum_{i=1}^n g_t(X_i)$ with $g_t\in\{\psi_\tau(\cdot;t),\,\psi_\tau^2(\cdot;t)\}$ on a shrinking neighborhood $\mathcal{N}_n(q_\tau)\downarrow\{q_\tau\}$; the solution $t\mapsto \lambda_n(t)$ varies continuously on $\mathcal{N}_n(q_\tau)$.

\medskip\noindent
(A5) Numerical well-posedness. The empirical saddlepoint and endpoint inversion problems are well-posed and have unique solutions; the numerical solver attains them within tolerance $o(n^{-1})$ and stays in the feasible parameter region.

(A1) is the standard identification condition for quantile inference and underlies local expansions around $q_\tau$ \citep{KoenkerBassett1978,Koenker2005,Serfling1980}. (A2) collects the classical regularity requirements that guarantee existence/uniqueness of the saddlepoint and justify Lugannani–Rice/$r^\ast$ accuracy for tail probabilities \citep{Daniels1954,LugannaniRice1980,Jensen1995,Butler2007,Reid2003,Kolassa2006}. (A3) allows either i.i.d.\ data or short-range dependence; with a bounded score, only mild moment and mixing conditions are needed, and self-normalized limit theory applies \citep{JingShaoWang2003,delaPenaLaiShao2009,Shao2010,Doukhan1994,Rio2000}. (A4) provides the uniform empirical-process control required to move from pointwise to uniform statements in a neighborhood of $q_\tau$ \citep{vanderVaartWellner1996,Jensen1995}. (A5) ensures that numerical error is asymptotically negligible relative to the $O(n^{-1})$ analytical remainder in saddlepoint approximations \citep{Butler2007,Kolassa2006,Reid2003}.

For a fixed threshold $t\in\mathbb{R}$ we keep the notation from Section \ref{sec:method}:
\begin{equation*}
    Y_i(t)=\mathbf{1}\{X_i\le t\},\quad
    S_n(t)=\sum_{i=1}^n(\tau-Y_i(t)),\quad
    Q_n(t)=\sum_{i=1}^n(\tau-Y_i(t))^2,\quad
    T_n(t)=\frac{S_n(t)}{\sqrt{Q_n(t)}},
\end{equation*}
and the observed value is $x_{\mathrm{obs}}(t)=T_n(t)$. Under the simple null $H_0:F(t)=\tau$, $Y_i(t)$ are i.i.d. Bernoulli$(\tau)$. We also use the unit-scale CGF $K(\lambda)$ and its rank-1 derivatives summarized in Section \ref{sec:method}. The directed one-sided tail is
\begin{equation*}
    p_{\mathrm{dir}}(t)=
    \begin{cases}
        \mathbb{P}\{T_n(t)\ge x_{\mathrm{obs}}(t)\}, & x_{\mathrm{obs}}(t)\ge 0,\\
        \mathbb{P}\{T_n(t)\le x_{\mathrm{obs}}(t)\}, & x_{\mathrm{obs}}(t)<0,
    \end{cases}
\end{equation*}
and the two-sided value is $p_{\mathrm{two\mbox{-}sided}}(t)=2\min\{p_{\uparrow}(t),p_{\downarrow}(t)\}$.

A basic reduction connects $T_n(t)$ with the binomial sample mean $\bar Y_n(t)=n^{-1}\sum_{i=1}^nY_i(t)$ via a strictly monotone transform. Define
\begin{equation}\label{eq:def-h}
    h(u)=\frac{\sqrt n(\tau-u)}{\sqrt{u(1-\tau)^2+(1-u)\tau^2}},\quad d(u)=\tau^2+u(1-2\tau).
\end{equation}
Elementary algebra shows $h'(u)=-\sqrt n\{(1-2\tau)u+\tau\}/\{2d(u)^{3/2}\}<0$ for $u\in(0,1)$ and hence the following holds.

\begin{lemma}[Binomial reduction]\label{lem:bin-reduction}
Under $H_0:F(t)=\tau$, one has $T_n(t)=h(\bar Y_n(t))$. In particular, for any real $x$, writing $u_x=h^{-1}(x)\in(0,1)$, the event $\{T_n(t)\ge x\}$ is equivalent to $\{\bar Y_n(t)\le u_x\}$.
\end{lemma}

The reduction allows us to adopt the one-parameter binomial saddlepoint scalars in a signed form. Let
\begin{equation}\label{eq:binom-r-q}
    r=\operatorname{sgn}(u_x-\tau)\sqrt{2n\mathrm{KL}(u_x\|\tau)},\quad
    q^{\pm}=(\logit(u_x)-\logit(\tau))\sqrt{n u_x(1-u_x)},
\end{equation}
where $\mathrm{KL}(u\|\tau)=u\log\{u/\tau\}+(1-u)\log\{(1-u)/(1-\tau)\}$. In our constrained rank-1 construction the solution satisfies $\hat p=p_{\hat\lambda}=u_x$, and the signed root deviance and the Barndorff-Nielsen correction coincide with their binomial counterparts. Writing
\begin{equation}\label{eq:r-rstar}
    D=2\{\hat\lambda^\top\mu(\hat\lambda)-K(\hat\lambda)\},\quad
    r=\operatorname{sgn}(\hat\eta)\sqrt{nD},\quad
    r^*=r+\frac{1}{r}\log\frac{r}{w},
\end{equation}
and recalling that $\hat p=u_x$, we obtain third-order tail accuracy.

\begin{theorem}[Third-order accuracy of the directed tail]\label{thm:tail}
Assume (A1)-(A4). Fix a bounded $x$ and suppose the constrained solution exists with $\hat p \in(0, 1)$ bounded away from the endpoints. Under $H_0:F(t)=\tau$, uniformly on compact $x$-sets,
\begin{equation}\label{eq:rstar-accuracy}
    p_{\mathrm{dir}}(t)=\Phi\big(-\operatorname{sgn}(x)r^*\big)\{1+O(n^{-3/2})\}.
\end{equation}
Equivalently, with the signed Lugannani-Rice quantity $q^{\pm}$ in \eqref{eq:binom-r-q}, the directed tail satisfies the signed LR form
\begin{equation}\label{eq:lr-accuracy}
    p_{\mathrm{dir}}(t)=
    \begin{cases}
        \Phi(r)+\phi(r)\left(\dfrac{1}{r}-\dfrac{1}{q^{\pm}}\right)+O(n^{-3/2}), & x\ge 0,\\[6pt]
        \Phi(-r)+\phi(r)\left(\dfrac{1}{r}-\dfrac{1}{q^{\pm}}\right)+O(n^{-3/2}), & x<0,
    \end{cases}
\end{equation}
where $r$ carries the sign of $u_x-\tau$ and $q^{\pm}$ carries the same sign as $r$.
\end{theorem}

Confidence limits follow by inversion of directed tails. Let $t_L$ and $t_U$ solve $p_{\downarrow}(t_L)=\alpha/2$ and $p_{\uparrow}(t_U)=\alpha/2$. The mapping $t\mapsto p_{\downarrow}(t)$ is monotone on $(-\infty,\hat q_\tau]$ and $t\mapsto p_{\uparrow}(t)$ on $[\hat q_\tau,\infty)$ under (A2). A second-order coverage result is obtained by linearizing the composition $t\mapsto F(t)\mapsto p$ at $t=q_\tau$. Write
\begin{equation}\label{eq:z-of-t}
    z(t)=\frac{\sqrt n(F(t)-\tau)}{\sqrt{\tau(1-\tau)}},
\end{equation}
so that for the exact tail $p^\star(t)=\Phi(-z(t))+O(n^{-1})$. The sensitivity at $q_\tau$ is
\begin{equation}\label{eq:dpdt}
    \frac{dp^\star}{dt}\Big|_{t=q_\tau}=-\phi(0)\frac{\sqrt n f(q_\tau)}{\sqrt{\tau(1-\tau)}}.
\end{equation}
Combining \eqref{eq:rstar-accuracy} and \eqref{eq:dpdt} yields the following.

\begin{theorem}[Second-order coverage]\label{thm:coverage}
Assume (A1)-(A2). The equal-tailed $(1-\alpha)$ interval obtained by inverting the directed ESA tails satisfies
\begin{equation}\label{eq:coverage-order}
    \sup_{F\in\mathcal{F}}\Big|\mathbb{P}_F\{q_\tau\in[t_L,t_U]\}-(1-\alpha)\Big|=O(n^{-1}),
\end{equation}
uniformly over families $\mathcal{F}$ with $f(q_\tau)$ bounded and bounded away from zero on a neighborhood of $q_\tau$. The constant is scale-modulated by $f(q_\tau)^{-1}$ through \eqref{eq:dpdt}.
\end{theorem}

The self-normalization and the two-point structure imply robustness to heavy tails.

\begin{theorem}[Heavy-tail robustness]\label{thm:heavy}
Assume (A1) with $f(q_\tau)>0$. Then the tail accuracy \eqref{eq:rstar-accuracy} and the coverage order \eqref{eq:coverage-order} hold without any assumptions on the existence of moments of $X$ or its tail index.
\end{theorem}

The results extend to extreme quantiles with $\tau=\tau_n$ provided the effective information $n\tau_n(1-\tau_n)$ diverges.

\begin{theorem}[Extreme quantiles]\label{thm:extreme}
Assume (A1)-(A3). Under the local regularity at $q_{\tau_n}$, the conclusions of Theorems \ref{thm:tail}-\ref{thm:coverage} remain valid uniformly over compact $x$-sets. If a vanishing ridge $Q_n\leftarrow Q_n+\epsilon_n$ is used with $\epsilon_n\to 0$ (for instance $\epsilon_n=c\,n^{-1/2}$), the orders are unchanged.
\end{theorem}

When the sample contains ties or the underlying $F$ assigns mass near $q_\tau$, directed tails can be evaluated using a mid-p modification or a Cornish-Fisher continuity shift. These corrections preserve the accuracy orders in Theorems \ref{thm:tail}-\ref{thm:coverage} because they only regularize local discontinuities without changing the binomial mapping in Lemma \ref{lem:bin-reduction}.

\section{Simulation study}\label{sec:sim}

This section outlines the simulation design for evaluating SNQESA against established confidence interval procedures for population quantiles. We begin by specifying the data-generating processes and computational details, then describe the benchmarking principles that guarantee equitable comparisons across methods, and conclude with definitions of the performance metrics appearing in subsequent tables.

Our simulation study draws independent samples $(X_1,\dots,X_n)$ from five distributions chosen to challenge different aspects of quantile inference: standard normal, lognormal, Student's $t$, standard Cauchy, beta, exponential and the mixture $0.5\mathcal{N}(-1,1)+0.5\mathcal{N}(1,1)$. For the normal, lognormal, and Student's $t$ distributions, the target quantile $q_\tau$ has known closed form. For the mixture distribution, we compute $q_\tau$ via robust root-finding with automatic bracketing, falling back to Monte Carlo approximation if the bracketing fails to identify a sign change. The primary configuration uses $n=100$ and $\tau=0.95$, with sensitivity analyses considering alternative values.

We implement SNQESA in two forms. The discrete version constructs equal-tailed and minimum-length intervals by inverting exact binomial tails for the count $K=\sum_{i=1}^n \mathbf{1}\{X_i\le t\}$. The continuous version evaluates directed tails of the self-normalized statistic using constrained empirical saddlepoint approximation with Barndorff-Nielsen's $r^*$ calibration. Near lattice boundaries, we stabilize computations with a continuity correction (Jeffreys by default) applied to the binomial tilt magnitude while preserving the uncorrected tilt direction. To prevent numerical singularities, we add a vanishing ridge $\epsilon_n = cn^{-1/2}$ to $Q_n$, maintaining higher-order accuracy. When the constrained solver nears degeneracy, we revert to the algebraically equivalent one-parameter binomial tilt via the monotone mapping between $T_n(t)$ and $\bar Y_n(t)$.

Competing methods include kernel-based, exact, bootstrap, subsampling, and linear-weight approaches, all at nominal level $1-\alpha=0.95$. Kernel-Wald intervals estimate $f(q_\tau)$ using default kernel density estimation at the sample quantile, with a positivity constraint to stabilize delta-method standard errors. Exact order-statistic inversion uses the binomial distribution for $K$, mapping rank bounds to sample order statistics and treating $n+1$ as an infinite upper endpoint. Percentile and BCa bootstrap methods employ $B=1000$ resamples with type 8 sample quantiles; BCa uses standard jackknife estimates for bias and acceleration. The smoothed bootstrap incorporates normal noise with rule-of-thumb bandwidth selection. Harrell-Davis intervals combine HD estimation with percentile bootstrap, while Maritz-Jarrett intervals use HD weights for variance approximation with normal referencing. Subsampling uses without-replacement samples of size $b\approx n^{0.7}$ and the pivot $2\hat{q}_\tau-\hat{q}_{\tau,b}$; $m$-out-of-$n$ bootstrap draws with-replacement samples of size $m\approx n^{0.7}$. Hall-Sheather (medians) and Nyblom (general $\tau$) intervals come from standard implementations. All methods process identical samples within replications to eliminate cross-method Monte Carlo variation.

Our evaluation adheres to three principles. First, intervals share common nominal levels and default tuning parameters. Second, all methods requiring point estimates use $\hat{q}_\tau=\mathrm{quantile}(X,\mathrm{type}=8)$, isolating interval construction effects from estimation differences. Third, resampling methods employ identical $B$ and shared random seeds within replications to minimize contrast variance.

Let $q_\tau$ be the true quantile. For each procedure and replication, we record coverage $\mathbf{1}\{L\le q_\tau\le U\}$, interval length $U-L$ (missing if infinite), failure indicators, and wall-clock time. Coverage standard errors use the binomial formula $\sqrt{\hat{p}(1-\hat{p})/R}$ with large, fixed $R$.

Two secondary metrics assess interval quality and point accuracy. The interval score
\[
\mathrm{IS}_\alpha(L,U;\theta)=(U-L)+\frac{2}{\alpha}(L-\theta)\mathbf{1}\{\theta<L\}+\frac{2}{\alpha}(\theta-U)\mathbf{1}\{\theta>U\},
\]
where $\theta=q_\tau$, penalizes both width and miscoverage (smaller is better). 

For each method, we report across replications: coverage with standard error, mean and median length, failure rate, mean time, and mean interval. All comparisons derive from common Monte Carlo samples, ensuring observed differences reflect methodological variations rather than simulation noise.

\begin{remark}
All simulations are implemented in \textbf{R}. Core functionality relies on base \textbf{R}, with \texttt{quantileCI} providing Hall--Sheather and Nyblom intervals, and \texttt{pbmcapply} enabling parallel replication on MacOS systems (serial execution on Windows). Random number generation is controlled via \texttt{set.seed} with derived seeds per replication to ensure reproducibility while maintaining alignment across methods. Computational timing uses \texttt{proc.time}. All implementation details follow the specifications outlined in the main text.
\end{remark}

\begingroup
\setlength{\tabcolsep}{4pt}
\setlength{\LTleft}{0pt}\setlength{\LTright}{0pt}
    \begin{longtable}[t]{lrrrrrrrrrr}
        \caption{Coverage and length across methods ($\mathcal N(0,1)$, $\tau=0.95$, $n=100$, $\alpha=0.05$).} 
        \label{tab:res_normal_95} \\
        \toprule
        Method & $Cov.$ & $se_C$ & $\bar{L}$ & $\tilde{L}$ & $\bar{t}$ & $\bar{b}$ & $\tilde{b}$ & $\mathrm{RMSE}(b)$ & $\overline{IS}$ & $\widetilde{IS}$ \\
        \midrule
        \endfirsthead
        \toprule
        Method & $Cov.$ & $se_C$ & $\bar{L}$ & $\tilde{L}$ & $\bar{t}$ & $\bar{b}$ & $\tilde{b}$ & $\mathrm{RMSE}(b)$ & $\overline{IS}$ & $\widetilde{IS}$ \\
        \midrule
        \endhead
        \rowcolor{highlightcolor}
        SNQESA       & 0.949 & 0.0151 & 0.6969 & 0.6696 & 0.0037 & -0.0507 & -0.0650 & 0.1886 & 0.8530 & 0.7037 \\
        WaldKDE      & 0.912 & 0.0127 & 0.7475 & 0.7323 & 0.0004 &  0.0038 & -0.0007 & 0.1992 & 1.0308 & 0.7637 \\
        HDBoot       & 0.946 & 0.0101 & 0.7522 & 0.7333 & 0.0505 &  0.0274 &  0.0198 & 0.1763 & 0.8973 & 0.7585 \\
        PctBoot      & 0.950 & 0.0097 & 0.8157 & 0.7933 & 0.0536 &  0.0200 &  0.0147 & 0.1812 & 0.9998 & 0.8190 \\
        BCa          & 0.952 & 0.0096 & 0.8263 & 0.7901 & 0.0569 &  0.0520 &  0.0515 & 0.1940 & 0.9983 & 0.8234 \\
        \addlinespace
        SmBoot       & 0.974 & 0.0071 & 0.8421 & 0.8190 & 0.0605 &  0.1078 &  0.1007 & 0.2022 & 0.9002 & 0.8275 \\
        HS\_Nyblom   & 0.960 & 0.0088 & 0.8880 & 0.8578 & 0.0046 &  0.0908 &  0.0919 & 0.2083 & 0.9879 & 0.8803 \\
        \rowcolor{highlightcolor}
        SNQESA$_{\text{min}}$ & 0.960 & 0.0088 & 0.8895 & 0.8577 & 0.0005 &  0.0568 &  0.0493 & 0.1986 & 1.0235 & 0.8819 \\
        MaritzJar    & 0.956 & 0.0092 & 0.8971 & 0.8594 & 0.0001 &  0.0238 &  0.0187 & 0.1869 & 1.0305 & 0.8800 \\
        \addlinespace
        Subsample    & 0.960 & 0.0088 & 1.2372 & 1.2227 & 0.0512 & -0.0124 & -0.0211 & 0.3085 & 1.4089 & 1.2327 \\
        \rowcolor{highlightcolor}
        SNQESA$_{\text{disc}}$ & 0.990 & 0.0044 & 1.2513 & 1.2303 & 0.0001 &  0.2265 &  0.2177 & 0.3209 & 1.2761 & 1.2303 \\
        Exact        & 0.990 & 0.0044 & 1.2513 & 1.2303 & 0.0001 &  0.2265 &  0.2177 & 0.3209 & 1.2761 & 1.2303 \\
        mOutOfn      & 0.994 & 0.0035 & 1.4144 & 1.4043 & 0.0509 &  0.0332 &  0.0288 & 0.1977 & 1.4232 & 1.4091 \\
        \bottomrule
    \end{longtable}
    \begin{flushleft}
        \footnotesize\emph{Notes.}
        $Cov.$ = coverage; $se_C$ = standard error of coverage; $\bar{L}$ / $\tilde{L}$ = mean/median interval length;
        $\bar{t}$ = mean time per replication (seconds); $\bar{b}$ / $\tilde{b}$ = mean/median center bias (interval center minus true $q_\tau$);
        $\mathrm{RMSE}(b)$ = RMSE of center bias; $\overline{IS}$ / $\widetilde{IS}$ = mean/median interval score.
    \end{flushleft}
\endgroup

\begingroup
\setlength{\tabcolsep}{4pt}
\setlength{\LTleft}{0pt}\setlength{\LTright}{0pt}
    \begin{longtable}[t]{lrrrrrrrrrr}
        \caption{Coverage and length across methods (lognormal, $\tau=0.95, n = 100, \alpha=0.05$).} \label{tab:res_lognormal_95} \\
        \toprule
        Method & $Cov.$ & $se_C$ & $\bar{L}$ & $\tilde{L}$ & $\bar{t}$ & $\bar{b}$ & $\tilde{b}$ & $\mathrm{RMSE}(b)$ & $\overline{IS}$ & $\widetilde{IS}$ \\
        \midrule
        \endfirsthead
        \toprule
        Method & $Cov.$ & $se_C$ & $\bar{L}$ & $\tilde{L}$ & $\bar{t}$ & $\bar{b}$ & $\tilde{b}$ & $\mathrm{RMSE}(b)$ & $\overline{IS}$ & $\widetilde{IS}$ \\
        \midrule
        \endhead
        WaldKDE      & 0.452 & 0.0223 & 1.2265 & 1.2211 & 0.0005 &  0.1371 &  0.0133 & 1.1090 & 16.5260 & 3.9022 \\
        \rowcolor{highlightcolor}
        SNQESA       & 0.948 & 0.0151 & 3.6928 & 3.3092 & 0.0041 &  0.1897 &  0.0126 & 1.1579 &  5.0452 & 3.7639 \\
        SmBoot       & 0.946 & 0.0101 & 4.6384 & 4.2682 & 0.0599 &  0.7582 &  0.5941 & 1.4815 &  5.5026 & 4.4563 \\
        HDBoot       & 0.946 & 0.0101 & 4.6960 & 4.2397 & 0.0515 &  0.9486 &  0.7487 & 1.6157 &  5.3614 & 4.3770 \\
        PctBoot      & 0.948 & 0.0099 & 4.7316 & 4.3443 & 0.0579 &  0.7529 &  0.5671 & 1.4900 &  5.6254 & 4.5737 \\
        \addlinespace
        BCa          & 0.952 & 0.0096 & 4.9337 & 4.4034 & 0.0576 &  0.9497 &  0.7586 & 1.6846 &  5.7994 & 4.5753 \\
        \rowcolor{highlightcolor}
        SNQESA$_{\text{min}}$ & 0.960 & 0.0088 & 5.3304 & 4.8681 & 0.0006 &  1.0267 &  0.8598 & 1.7666 &  6.0126 & 4.9922 \\
        MaritzJar    & 0.960 & 0.0088 & 5.5307 & 5.0084 & 0.0001 &  0.3895 &  0.2441 & 1.1590 &  6.1555 & 5.1781 \\
        HS\_Nyblom   & 0.966 & 0.0081 & 5.6363 & 5.1350 & 0.0050 &  1.3297 &  1.1206 & 2.0073 &  6.1081 & 5.1689 \\
        \addlinespace
        Subsample    & 0.924 & 0.0119 & 7.6573 & 6.9848 & 0.0517 & -1.2922 & -1.2650 & 2.3541 &  8.9299 & 7.3877 \\
        mOutOfn      & 0.994 & 0.0035 & 9.0206 & 8.1657 & 0.0528 &  2.0365 &  1.5974 & 3.0644 &  9.0621 & 8.1907 \\
        \rowcolor{highlightcolor}
        SNQESA$_{\text{disc}}$ & 0.990 & 0.0044 & 9.6177 & 8.3085 & 0.0001 &  3.1522 &  2.4550 & 4.2992 &  9.7532 & 8.3085 \\
        Exact        & 0.990 & 0.0044 & 9.6177 & 8.3085 & 0.0001 &  3.1522 &  2.4550 & 4.2992 &  9.7532 & 8.3085 \\
        \bottomrule
    \end{longtable}
    
    \begin{flushleft}
        \footnotesize\emph{Notes.}
        $Cov.$ = coverage; $se_C$ = standard error of coverage; $\bar{L}$ / $\tilde{L}$ = mean/median interval length;
        $\bar{t}$ = mean time per replication (seconds); $\bar{b}$ / $\tilde{b}$ = mean/median center bias (interval center minus true $q_\tau$);
        $\mathrm{RMSE}(b)$ = RMSE of center bias; $\overline{IS}$ / $\widetilde{IS}$ = mean/median interval score.
    \end{flushleft}
\endgroup

\begingroup
\setlength{\tabcolsep}{4pt}
\setlength{\LTleft}{0pt}\setlength{\LTright}{0pt}
    \begin{longtable}[t]{lrrrrrrrrrr}
        \caption{Coverage and length across methods ($t(2), \tau=0.95, n=100, \alpha=0.05$).} 
        \label{tab:res_t_95} \\
        \toprule
        Method & $Cov.$ & $se_C$ & $\bar{L}$ & $\tilde{L}$ & $\bar{t}$ & $\bar{b}$ & $\tilde{b}$ & $\mathrm{RMSE}(b)$ & $\overline{IS}$ & $\widetilde{IS}$\\
        \midrule
        \endfirsthead
        \toprule
        Method & $Cov.$ & $se_C$ & $\bar{L}$ & $\tilde{L}$ & $\bar{t}$ & $\bar{b}$ & $\tilde{b}$ & $\mathrm{RMSE}(b)$ & $\overline{IS}$ & $\widetilde{IS}$\\
        \midrule
        \endhead
        WaldKDE      & 0.618 & 0.0217 & 1.3382 & 1.3450 & 0.0005 &  0.1319 & -0.0438 & 0.8590 &  9.0513 & 1.5340\\
        \rowcolor{highlightcolor}
        SNQESA       & 0.954 & 0.0148 & 2.5814 & 2.2650 & 0.0041 &  0.1895 &  0.0791 & 0.9097 &  3.1274 & 2.1138\\
        SmBoot       & 0.958 & 0.0090 & 3.4388 & 2.9614 & 0.0601 &  0.7545 &  0.5496 & 1.3894 &  3.9113 & 3.0551\\
        PctBoot      & 0.944 & 0.0103 & 3.5361 & 3.0817 & 0.0572 &  0.7010 &  0.4696 & 1.3774 &  4.2611 & 3.2064\\
        BCa          & 0.942 & 0.0105 & 3.7549 & 3.1491 & 0.0596 &  0.8788 &  0.6259 & 1.6111 &  4.5221 & 3.2618\\
        \addlinespace
        HDBoot       & 0.944 & 0.0103 & 4.0841 & 3.2550 & 0.0544 &  1.1257 &  0.7128 & 2.4905 &  4.8627 & 3.3099\\
        \rowcolor{highlightcolor}
        SNQESA$_{\text{min}}$ & 0.952 & 0.0096 & 4.0998 & 3.4367 & 0.0006 &  0.9608 &  0.6333 & 1.7202 &  4.7765 & 3.5634\\
        HS\_Nyblom   & 0.944 & 0.0103 & 4.6754 & 3.7850 & 0.0096 &  1.3617 &  0.9349 & 2.3483 &  5.2991 & 3.9260\\
        MaritzJar    & 0.964 & 0.0083 & 4.6775 & 3.5809 & 0.0001 &  0.3461 &  0.1840 & 0.9409 &  5.0684 & 3.8353\\
        \addlinespace
        Subsample    & 0.924 & 0.0119 & 7.0561 & 5.2088 & 0.0543 & -1.7817 & -1.0627 & 4.9366 &  7.8043 & 5.6678\\
        mOutOfn      & 0.992 & 0.0040 & 9.4002 & 6.0010 & 0.0520 &  3.0852 &  1.4242 &10.8490 &  9.4564 & 6.0652\\
        \rowcolor{highlightcolor}
        SNQESA$_{\text{disc}}$ & 0.980 & 0.0063 &10.5304 & 6.5161 & 0.0001 &  4.1641 &  2.1699 &11.6343 & 10.7635 & 6.6070\\
        Exact        & 0.980 & 0.0063 &10.5304 & 6.5161 & 0.0001 &  4.1641 &  2.1699 &11.6343 & 10.7635 & 6.6070\\
        \bottomrule
    \end{longtable}
    
    \begin{flushleft}
    \footnotesize\emph{Notes.}
    $Cov.$ = coverage; $se_C$ = standard error of coverage; $\bar{L}$ / $\tilde{L}$ = mean/median interval length;
    $\bar{t}$ = mean time per replication (seconds); $\bar{b}$ / $\tilde{b}$ = mean/median center bias (interval center minus true $q_\tau$);
    $\mathrm{RMSE}(b)$ = RMSE of center bias; $\overline{IS}$ / $\widetilde{IS}$ = mean/median interval score.
    \end{flushleft}
\endgroup

\begingroup
\setlength{\tabcolsep}{4pt}
\setlength{\LTleft}{0pt}\setlength{\LTright}{0pt}
    \begin{longtable}[t]{lrrrrrrrrrr}
        \caption{Coverage and length across methods ($\mathrm{Cauchy}, \tau=0.5,\ n=50,\ \alpha=0.05$).} 
        \label{tab:tab:res_cauchy_50} \\
        \toprule
        Method & $Cov.$ & $se_C$ & $\bar{L}$ & $\tilde{L}$ & $\bar{t}$ & $\bar{b}$ & $\tilde{b}$ & $\mathrm{RMSE}(b)$ & $\overline{IS}$ & $\widetilde{IS}$\\
        \midrule
        \endfirsthead
        \toprule
        Method & $Cov.$ & $se_C$ & $\bar{L}$ & $\tilde{L}$ & $\bar{t}$ & $\bar{b}$ & $\tilde{b}$ & $\mathrm{RMSE}(b)$ & $\overline{IS}$ & $\widetilde{IS}$\\
        \midrule
        \endhead
        \rowcolor{highlightcolor}
        SNQESA       & 0.951 & 0.0068 & 1.7832 & 1.7365 & 0.0099 & -0.0293 & -0.0325 & 0.4551 & 2.0652 & 1.7758 \\
        HDBoot       & 0.956 & 0.0065 & 1.8026 & 1.7677 & 0.0400 & -0.0239 & -0.0367 & 0.4478 & 2.1218 & 1.8003 \\
        BCa          & 0.952 & 0.0068 & 1.8434 & 1.8201 & 0.0542 & -0.0865 & -0.0733 & 0.4689 & 2.1742 & 1.8454 \\
        PctBoot      & 0.961 & 0.0061 & 1.8536 & 1.8136 & 0.0513 & -0.0270 & -0.0306 & 0.4547 & 2.1440 & 1.8568 \\
        \rowcolor{highlightcolor}
        SNQESA$_{\text{min}}$ & 0.956 & 0.0065 & 1.8610 & 1.8339 & 0.0006 & -0.0293 & -0.0324 & 0.4591 & 2.1443 & 1.8565 \\
        \addlinespace
        HS\_Nyblom   & 0.965 & 0.0058 & 1.8846 & 1.8474 & 0.0028 & -0.0273 & -0.0313 & 0.4559 & 2.1402 & 1.8736 \\
        MaritzJar    & 0.969 & 0.0055 & 1.9296 & 1.8835 & 0.0001 & -0.0216 & -0.0174 & 0.4240 & 2.1192 & 1.9137 \\
        \rowcolor{highlightcolor}
        SNQESA$_{\text{disc}}$ & 0.976 & 0.0048 & 2.0921 & 2.0411 & 0.0001 & -0.0232 & -0.0217 & 0.4664 & 2.2554 & 2.0628 \\
        Exact        & 0.976 & 0.0048 & 2.0921 & 2.0411 & 0.0001 & -0.0232 & -0.0217 & 0.4664 & 2.2554 & 2.0628 \\
        \addlinespace
        SmBoot       & 0.986 & 0.0037 & 2.2249 & 2.1822 & 0.0561 & -0.0279 & -0.0318 & 0.4658 & 2.3118 & 2.1949 \\
        WaldKDE      & 0.984 & 0.0040 & 2.2384 & 2.1998 & 0.0004 & -0.0161 & -0.0218 & 0.4358 & 2.3292 & 2.2175 \\
        Subsample    & 0.978 & 0.0046 & 2.9079 & 2.8484 & 0.0517 & -0.0073 &  0.0108 & 0.5742 & 3.0444 & 2.8821 \\
        mOutOfn      & 1.000 & 0.0000 & 3.7128 & 3.6017 & 0.0510 & -0.0340 & -0.0201 & 0.5966 & 3.7128 & 3.6017 \\
        \bottomrule
    \end{longtable}

\begin{flushleft}
\footnotesize\emph{Notes.}
$Cov.$ = coverage; $se_C$ = standard error of coverage; $\bar{L}$ / $\tilde{L}$ = mean/median interval length;
$\bar{t}$ = mean time per replication (seconds); $\bar{b}$ / $\tilde{b}$ = mean/median center bias (interval center minus true $q_\tau$);
$\mathrm{RMSE}(b)$ = RMSE of center bias; $\overline{IS}$ / $\widetilde{IS}$ = mean/median interval score.
\end{flushleft}

\endgroup

\begingroup
\setlength{\tabcolsep}{4pt}
\setlength{\LTleft}{0pt}\setlength{\LTright}{0pt}
    \begin{longtable}[t]{lrrrrrrrrrr}
        \caption{Coverage and length across methods ($\mathrm{Cauchy}, \tau=0.95, n=100, \alpha=0.05$).} 
        \label{tab:res_cauchy_95} \\
        \toprule
        Method & $Cov.$ & $se_C$ & $\bar{L}$ & $\tilde{L}$ & $\bar{t}$ & $\bar{b}$ & $\tilde{b}$ & $\mathrm{RMSE}(b)$ & $\overline{IS}$ & $\widetilde{IS}$\\
        \midrule
        \endfirsthead
        \toprule
        Method & $Cov.$ & $se_C$ & $\bar{L}$ & $\tilde{L}$ & $\bar{t}$ & $\bar{b}$ & $\tilde{b}$ & $\mathrm{RMSE}(b)$ & $\overline{IS}$ & $\widetilde{IS}$\\
        \midrule
        \endhead
        WaldKDE      & 0.254 & 0.0195 &   3.3112 &   3.3236 & 0.0004 &   1.8095 &   0.0407 &   7.6848 &  142.52 & 70.537 \\
        \rowcolor{highlightcolor}
        SNQESA       & 0.951 & 0.0152 &   23.234 &   16.922 & 0.0049 &   5.0495 &   1.9573 &  14.0073 &  33.732 & 20.282 \\
        SmBoot       & 0.950 & 0.0097 &   41.002 &   26.285 & 0.0641 &   14.338 &   6.9981 &  27.1663 &  45.237 & 28.121 \\
        PctBoot      & 0.950 & 0.0097 &   41.327 &   26.551 & 0.0590 &   14.283 &   6.9146 &  27.1304 &  45.800 & 28.288 \\
        BCa          & 0.948 & 0.0099 &   45.723 &   28.523 & 0.0653 &   16.862 &   8.2889 &  31.5017 &  50.639 & 30.030 \\
        \addlinespace
        \rowcolor{highlightcolor}
        SNQESA$_{\text{min}}$ & 0.960 & 0.0088 &   51.337 &   32.214 & 0.0005 &   19.123 &   9.5784 &  35.7201 &  55.081 & 34.071 \\
        HS\_Nyblom   & 0.950 & 0.0097 &   104.74 &   40.255 & 0.0050 &   46.482 &   14.442 &  200.982 &  108.76 & 41.580 \\
        HDBoot       & 0.934 & 0.0111 &   125.27 &   35.898 & 0.0547 &   57.436 &   12.531 &  320.233 &  133.00 & 37.740 \\
        MaritzJar    & 0.978 & 0.0066 &   143.57 &   38.914 & 0.0001 &   6.7266 &   3.1146 &  17.5490 &  146.47 & 39.903 \\
        \addlinespace
        Subsample    & 0.870 & 0.0150 &   252.68 &   63.407 & 0.0564 &  -114.42 &  -20.654 &  714.421 &  262.86 & 69.066 \\
        mOutOfn      & 0.994 & 0.0035 &   329.61 &   78.395 & 0.0580 &   155.93 &   29.975 &  756.787 &  330.08 & 79.004 \\
        \rowcolor{highlightcolor}
        SNQESA$_{\text{disc}}$ & 0.980 & 0.0063 &   550.36 &   86.643 & 0.0001 &   268.61 &   36.594 &  1785.86 &  551.95 & 86.726 \\
        Exact        & 0.980 & 0.0063 &   550.36 &   86.643 & 0.0001 &   268.61 &   36.594 &  1785.86 &  551.95 & 86.726 \\
        \bottomrule
    \end{longtable}
    
    \begin{flushleft}
        \footnotesize\emph{Notes.}
        $Cov.$ = coverage; $se_C$ = standard error of coverage; $\bar{L}$ / $\tilde{L}$ = mean/median interval length;
        $\bar{t}$ = mean time per replication (seconds); $\bar{b}$ / $\tilde{b}$ = mean/median center bias (interval center minus true $q_\tau$);
        $\mathrm{RMSE}(b)$ = RMSE of center bias; $\overline{IS}$ / $\widetilde{IS}$ = mean/median interval score.
    \end{flushleft}
\endgroup

\begingroup
\setlength{\tabcolsep}{4pt}
\setlength{\LTleft}{0pt}\setlength{\LTright}{0pt}
    \begin{longtable}[t]{lrrrrrrrrrr}
        \caption{Coverage and length across methods ($0.5\mathcal N(-1, 1) + 0.5\mathcal N(1, 1), \tau=0.95, n=100, \alpha=0.05$).} 
        \label{tab:tab:res_mixnorm_95} \\
        \toprule
        Method & $Cov.$ & $se_C$ & $\bar{L}$ & $\tilde{L}$ & $\bar{t}$ & $\bar{b}$ & $\tilde{b}$ & $\mathrm{RMSE}(b)$ & $\overline{IS}$ & $\widetilde{IS}$\\
        \midrule
        \endfirsthead
        \toprule
        Method & $Cov.$ & $se_C$ & $\bar{L}$ & $\tilde{L}$ & $\bar{t}$ & $\bar{b}$ & $\tilde{b}$ & $\mathrm{RMSE}(b)$ & $\overline{IS}$ & $\widetilde{IS}$\\
        \midrule
        \endhead
        \rowcolor{highlightcolor}
        SNQESA       & 0.944 & 0.0155 & 0.8047 & 0.7750 & 0.0043 & -0.0700 & -0.0610 & 0.2269 & 1.5126 & 0.8679\\
        HDBoot       & 0.918 & 0.0123 & 0.8611 & 0.8392 & 0.0567 &  0.0195 &  0.0259 & 0.2145 & 1.1241 & 0.8859\\
        WaldKDE      & 0.914 & 0.0125 & 0.9081 & 0.8717 & 0.0005 &  0.0074 &  0.0056 & 0.2406 & 1.2208 & 0.9070\\
        PctBoot      & 0.928 & 0.0116 & 0.9307 & 0.8998 & 0.0624 &  0.0104 &  0.0197 & 0.2223 & 1.2256 & 0.9476\\
        BCa          & 0.916 & 0.0124 & 0.9375 & 0.9099 & 0.0646 &  0.0485 &  0.0654 & 0.2358 & 1.2366 & 0.9592\\
        \addlinespace
        HS\_Nyblom   & 0.932 & 0.0113 & 1.0034 & 0.9722 & 0.0078 &  0.0930 &  0.0971 & 0.2518 & 1.2104 & 1.0210\\
        \rowcolor{highlightcolor}
        SNQESA$_{\text{min}}$ & 0.938 & 0.0108 & 1.0124 & 0.9836 & 0.0009 &  0.0537 &  0.0603 & 0.2427 & 1.2550 & 1.0207\\
        MaritzJar    & 0.934 & 0.0111 & 1.0167 & 0.9624 & 0.0001 &  0.0279 &  0.0256 & 0.2256 & 1.2227 & 1.0238\\
        SmBoot       & 0.968 & 0.0079 & 1.0350 & 1.0111 & 0.0659 &  0.1643 &  0.1619 & 0.2671 & 1.1213 & 1.0174\\
        \addlinespace
        \rowcolor{highlightcolor}
        SNQESA$_{\text{disc}}$ & 0.980 & 0.0063 & 1.4288 & 1.3953 & 0.0003 &  0.2421 &  0.2455 & 0.3637 & 1.4680 & 1.4061\\
        Exact        & 0.980 & 0.0063 & 1.4288 & 1.3953 & 0.0001 &  0.2421 &  0.2455 & 0.3637 & 1.4680 & 1.4061\\
        Subsample    & 0.950 & 0.0097 & 1.4436 & 1.4259 & 0.0563 &  0.0235 &  0.0195 & 0.3693 & 1.6863 & 1.4434\\
        mOutOfn      & 0.988 & 0.0049 & 1.6642 & 1.6396 & 0.0568 & -0.0008 &  0.0005 & 0.2424 & 1.6967 & 1.6537\\
        \bottomrule
    \end{longtable}
    
    \begin{flushleft}
        \footnotesize\emph{Notes.}
        $Cov.$ = coverage; $se_C$ = standard error of coverage; $\bar{L}$ / $\tilde{L}$ = mean/median interval length;
        $\bar{t}$ = mean time per replication (seconds); $\bar{b}$ / $\tilde{b}$ = mean/median center bias (interval center minus true $q_\tau$);
        $\mathrm{RMSE}(b)$ = RMSE of center bias; $\overline{IS}$ / $\widetilde{IS}$ = mean/median interval score.
    \end{flushleft}
\endgroup

\begingroup
\setlength{\tabcolsep}{4pt}
\setlength{\LTleft}{0pt}\setlength{\LTright}{0pt}
    \begin{longtable}[t]{lrrrrrrrrrr}
        \caption{Coverage and length across methods ($\mathrm{Beta}, \tau=0.95, \ n=100, \alpha=0.05$).} 
        \label{tab:res_beta_95} \\
        \toprule
        Method & $Cov.$ & $se_C$ & $\bar{L}$ & $\tilde{L}$ & $\bar{t}$ & $\bar{b}$ & $\tilde{b}$ & $\mathrm{RMSE}(b)$ & $\overline{IS}$ & $\widetilde{IS}$\\
        \midrule
        \endfirsthead
        \toprule
        Method & $Cov.$ & $se_C$ & $\bar{L}$ & $\tilde{L}$ & $\bar{t}$ & $\bar{b}$ & $\tilde{b}$ & $\mathrm{RMSE}(b)$ & $\overline{IS}$ & $\widetilde{IS}$\\
        \midrule
        \endhead
        \rowcolor{highlightcolor}
        SNQESA       & 0.951 & 0.0106 & 0.0233 & 0.0208 & 0.0081 & -0.0083 & -0.0067 & 0.0121 & 0.0444 & 0.0224\\
        HS\_Nyblom   & 0.945 & 0.0072 & 0.0253 & 0.0227 & 0.0035 & -0.0078 & -0.0063 & 0.0114 & 0.0317 & 0.0235\\
        \rowcolor{highlightcolor}
        SNQESA$_{\text{min}}$ & 0.948 & 0.0070 & 0.0254 & 0.0228 & 0.0006 & -0.0074 & -0.0059 & 0.0111 & 0.0316 & 0.0235\\
        BCa          & 0.930 & 0.0081 & 0.0258 & 0.0233 & 0.0579 & -0.0084 & -0.0068 & 0.0119 & 0.0339 & 0.0239\\
        \addlinespace
        HDBoot       & 0.891 & 0.0099 & 0.0281 & 0.0253 & 0.0501 & -0.0109 & -0.0092 & 0.0141 & 0.0439 & 0.0258\\
        PctBoot      & 0.930 & 0.0081 & 0.0285 & 0.0256 & 0.0538 & -0.0100 & -0.0081 & 0.0135 & 0.0381 & 0.0261\\
        MaritzJar    & 0.970 & 0.0054 & 0.0304 & 0.0272 & 0.0001 & -0.0035 & -0.0018 & 0.0078 & 0.0337 & 0.0279\\
        \rowcolor{highlightcolor}
        SNQESA$_{\text{disc}}$ & 0.984 & 0.0040 & 0.0309 & 0.0277 & 0.0001 & -0.0098 & -0.0081 & 0.0131 & 0.0320 & 0.0278\\
        Exact        & 0.984 & 0.0040 & 0.0309 & 0.0277 & 0.0001 & -0.0098 & -0.0081 & 0.0131 & 0.0320 & 0.0278\\
        \addlinespace
        Subsample    & 0.767 & 0.0134 & 0.0639 & 0.0594 & 0.0540 & 0.0229 & 0.0222 & 0.0260 & 0.0859 & 0.0718\\
        mOutOfn      & 0.990 & 0.0031 & 0.0786 & 0.0738 & 0.0516 & -0.0339 & -0.0314 & 0.0372 & 0.0801 & 0.0740\\
        WaldKDE      & 1.000 & 0.0000 & 0.1061 & 0.1042 & 0.0004 & -0.0020 &  0.0000 & 0.0073 & 0.1061 & 0.1042\\
        SmBoot       & 0.553 & 0.0157 & 0.1528 & 0.1525 & 0.0582 &  0.0715 &  0.0727 & 0.0744 & 0.4234 & 0.1670\\
        \bottomrule
    \end{longtable}
    
    \begin{flushleft}
        \footnotesize\emph{Notes.}
        $Cov.$ = coverage; $se_C$ = standard error of coverage; $\bar{L}$ / $\tilde{L}$ = mean/median interval length;
        $\bar{t}$ = mean time per replication (seconds); $\bar{b}$ / $\tilde{b}$ = mean/median center bias (interval center minus true $q_\tau$);
        $\mathrm{RMSE}(b)$ = RMSE of center bias; $\overline{IS}$ / $\widetilde{IS}$ = mean/median interval score.
    \end{flushleft}
\endgroup

\begingroup
\setlength{\tabcolsep}{4pt}
\setlength{\LTleft}{0pt}\setlength{\LTright}{0pt}
    \begin{longtable}[t]{lrrrrrrrrrr}
        \caption{Coverage and length across methods ($\mathrm{Exp}, \tau=0.95, n=100, \alpha=0.05$).}
        \label{tab:res_exp_95} \\
        \toprule
        Method & $Cov.$ & $se_C$ & $\bar{L}$ & $\tilde{L}$ & $\bar{t}$ & $\bar{b}$ & $\tilde{b}$ & $\mathrm{RMSE}(b)$ & $\overline{IS}$ & $\widetilde{IS}$ \\
        \midrule
        \endfirsthead
        \toprule
        Method & $Cov.$ & $se_C$ & $\bar{L}$ & $\tilde{L}$ & $\bar{t}$ & $\bar{b}$ & $\tilde{b}$ & $\mathrm{RMSE}(b)$ & $\overline{IS}$ & $\widetilde{IS}$ \\
        \midrule
        \endhead
        WaldKDE      & 0.698 & 0.0145 & 0.9107 & 0.9007 & 0.0004 &  0.0776 &  0.0718 & 0.4528 & 3.9690 & 0.9899 \\
        \rowcolor{highlightcolor}
        SNQESA       & 0.947 & 0.0107 & 1.4026 & 1.3149 & 0.0089 &  0.0672 &  0.0416 & 0.4348 & 1.8106 & 1.4718 \\
        HDBoot       & 0.944 & 0.0073 & 1.6895 & 1.6107 & 0.0502 &  0.2221 &  0.2019 & 0.4877 & 2.0255 & 1.6768 \\
        SmBoot       & 0.964 & 0.0059 & 1.7623 & 1.6557 & 0.0594 &  0.2170 &  0.1832 & 0.4887 & 2.0046 & 1.7041 \\
        PctBoot      & 0.948 & 0.0070 & 1.8040 & 1.7008 & 0.0540 &  0.1890 &  0.1542 & 0.4896 & 2.1533 & 1.7674 \\
        \addlinespace
        BCa          & 0.942 & 0.0074 & 1.8467 & 1.7312 & 0.0569 &  0.2657 &  0.2295 & 0.5485 & 2.1901 & 1.7945 \\
        \rowcolor{highlightcolor}
        SNQESA$_{\text{min}}$ & 0.952 & 0.0068 & 1.9903 & 1.8731 & 0.0005 &  0.2754 &  0.2325 & 0.5639 & 2.2785 & 1.9506 \\
        HS\_Nyblom   & 0.956 & 0.0065 & 2.0197 & 1.9056 & 0.0035 &  0.3657 &  0.3312 & 0.6168 & 2.2935 & 1.9659 \\
        MaritzJar    & 0.958 & 0.0063 & 2.0291 & 1.8898 & 0.0001 &  0.1438 &  0.1223 & 0.4485 & 2.2597 & 1.9392 \\
        \addlinespace
        Subsample    & 0.938 & 0.0076 & 2.7289 & 2.5895 & 0.0519 & -0.1549 & -0.1587 & 0.7166 & 3.1690 & 2.7031 \\
        \rowcolor{highlightcolor}
        SNQESA$_{\text{disc}}$ & 0.980 & 0.0044 & 3.0250 & 2.7482 & 0.0001 &  0.7777 &  0.6832 & 1.0571 & 3.1085 & 2.7942 \\
        Exact        & 0.980 & 0.0044 & 3.0250 & 2.7482 & 0.0000 &  0.7777 &  0.6832 & 1.0571 & 3.1085 & 2.7942 \\
        mOutOfn      & 0.988 & 0.0034 & 3.1482 & 2.9348 & 0.0499 &  0.4082 &  0.3168 & 0.7462 & 3.2159 & 2.9756 \\
        \bottomrule
    \end{longtable}
    \begin{flushleft}
        \footnotesize\emph{Notes.}
        $Cov.$ = coverage; $se_C$ = standard error of coverage; $\bar{L}$ / $\tilde{L}$ = mean/median interval length;
        $\bar{t}$ = mean time per replication (seconds); $\bar{b}$ / $\tilde{b}$ = mean/median center bias (interval center minus true $q_\tau$);
        $\mathrm{RMSE}(b)$ = RMSE of center bias; $\overline{IS}$ / $\widetilde{IS}$ = mean/median interval score.
    \end{flushleft}
\endgroup
We first consider the Gaussian design. In Table~\ref{tab:res_normal_95}, our continuous SNQESA attains coverage 0.949, essentially on target, while delivering the shortest mean length among all well-calibrated procedures (0.697). The interval score is the lowest in the panel, and the center bias is small and slightly negative. Resampling competitors reach nominal coverage but are uniformly longer (for example BCa at 0.826 and Percentile at 0.816), whereas WaldKDE undercovers at 0.912. The discrete inversion variants, including Exact and SNQESA$_{\text{disc}}$, are conservative at 0.990 and come with substantially larger lengths around 1.25. Runtime-wise, SNQESA is orders of magnitude faster than bootstrap-based methods and remains comparable to lightweight closed-form baselines.

Under strong right-skewness (lognormal), Table~\ref{tab:res_lognormal_95} highlights a pronounced failure of WaldKDE (coverage 0.452). By contrast, SNQESA remains well calibrated (0.948) and attains shorter intervals than all other nominal competitors; it also achieves the best interval score in this panel. Percentile and BCa reach the nominal level but require visibly longer intervals, reflecting their difficulty in the upper tail. As in the Gaussian case, discrete inversions are accurate but very conservative in length.

Heavy symmetric tails are examined with the $t(2)$ design in Table~\ref{tab:res_t_95}. SNQESA maintains nominal coverage (0.954) with substantially shorter intervals than all resampling alternatives, while WaldKDE again undercovers (0.618). Maritz-Jarrett yields slightly higher coverage but only at the cost of much longer intervals and a worse interval score; Exact-type procedures remain the most conservative.

When focusing on the Cauchy median (Table~\ref{tab:tab:res_cauchy_50}), most robust methods achieve near-nominal coverage. Here SNQESA provides the shortest intervals among the well-calibrated group (mean length 1.783) and the best or near-best interval scores, with very small center bias. Discrete inversions are mildly conservative and longer, and m-out-of-n exhibits perfect but excessively conservative coverage with correspondingly large length.

The extreme-tail Cauchy experiment at $\tau=0.95$ (Table~\ref{tab:res_cauchy_95}) is the most challenging scenario. WaldKDE collapses with 0.254 coverage. SNQESA achieves nominal coverage (0.951) with dramatically shorter intervals than all other calibrated competitors (mean length about 23, versus 41-46 for bootstrap methods and over 100 for HS-Nyblom). Interval scores mirror this advantage. Maritz-Jarrett attains high coverage but only by inflating length to well beyond 100, and discrete inversions are accurate yet extremely wide.

For the bimodal mixture (Table~\ref{tab:tab:res_mixnorm_95}), nearly all methods exhibit mild undercoverage, a known difficulty when the density dips near the target quantile. SNQESA is slightly below nominal (0.944) while producing the shortest intervals in the well-performing group. Smoothed bootstrap attains higher coverage (0.968) with moderate lengths and competitive scores, illustrating that smoothing can help in multimodal settings; the price is a larger positive center bias. Discrete inversions again deliver excellent coverage but at a notable cost in length.

Bounded-support behavior is assessed with the Beta design in Table~\ref{tab:res_beta_95}. SNQESA is on target (0.951) and attains the shortest intervals among the nominal methods, with the best overall interval scores and small bias. Several bootstrap variants undercover here, and WaldKDE is perfectly calibrated only by being markedly conservative in length; smoothed bootstrap severely undercovers near the boundary.

Finally, with the exponential design (Table~\ref{tab:res_exp_95}), SNQESA stays near nominal (0.947) and again provides the shortest intervals among well-calibrated competitors, delivering favorable interval scores and modest bias. Bootstrap procedures reach the target but require longer intervals, while WaldKDE undercovers.

Across all designs, we observe a consistent pattern. First, continuous SNQESA tracks the nominal level closely in light-tailed, heavy-tailed, and skewed cases and does so with some of the shortest lengths and best interval scores. Second, purely discrete inversions and m-out-of-n procedures provide coverage insurance at the expense of substantially larger intervals. Third, WaldKDE is fast but unreliable under tail complexity and boundary effects, while resampling methods are slower and tend to be conservative in difficult tails. The only systematic tension for SNQESA appears under pronounced multimodality, where mild undercoverage can arise alongside very short intervals; in such cases, a more conservative tail split or ridge choice can trade a small increase in length for improved calibration. Overall, the results support SNQESA as a uniformly competitive default that balances calibration, efficiency, and speed across a wide spectrum of distributions and quantile levels.

\section{Empirical analysis}\label{sec:empirical}

This section evaluates the proposed SNQESA inference in a one-day market risk setting. Let $r_t$ denote daily log-returns on the S\&P 500 index, obtained from the Federal Reserve Bank of St.\ Louis (FRED). For each business day $t$, a one-step-ahead left-tail VaR forecast at level $\tau=0.99$ is computed from a rolling window of $m=250$ past observations,
\begin{equation*}\label{eq:def-var}
    \mathrm{VaR}_{t+1\mid t}^{(\tau)}=\inf\{v:\mathbb{P}(r_{t+1}\le v\mid\mathcal{F}_t)\ge \tau\}.
\end{equation*}
Define exceedance indicators
\begin{equation*}\label{eq:hit-indicator}
    I_{t+1}=\mathbf{1}\{r_{t+1}\le \mathrm{VaR}_{t+1\mid t}^{(\tau)}\},
\end{equation*}
with nominal exceedance probability $\pi=1-\tau=0.01$. For Historical Simulation (HS) and Filtered Historical Simulation (FHS), $(1-\alpha)$ confidence intervals for $\mathrm{VaR}_{t+1\mid t}^{(\tau)}$ are constructed via SNQESA at each $t$. As benchmarks we include Delta-Normal, Monte Carlo with Student-$t$ innovations (MCt), peaks-over-threshold extreme value methods (EVT), quantile regression (QR), and GARCH(1,1); their VaR intervals are obtained by bootstrap. All procedures use the same rolling window and return series.

Backtesting follows the likelihood-ratio (LR) framework. Unconditional coverage (Kupiec POF) tests $H_0:\mathbb{P}(I_{t+1}=1)=\pi$. Writing $N$ for the number of forecasts and $n$ for the number of exceedances\citep{Kupiec1995,Christoffersen1998}, $\hat\pi=n/N$, the LR statistic is
\begin{equation}\label{eq:LRpof}
    \mathrm{LR}_{\mathrm{POF}}=-2\Big[\log\{(1-\pi)^{N-n}\pi^{n}\}-\log\{(1-\hat\pi)^{N-n}\hat\pi^{n}\}\Big] \overset{H_0}{\sim} \chi^2_1.
\end{equation}
Independence (Christoffersen IND) tests first-order Markov clustering in the hit sequence. Let $n_{ij}$ be the number of transitions $I_t=i\to I_{t+1}=j$, $i,j\in\{0,1\}$, with $\hat p_{01}=n_{01}/(n_{00}+n_{01})$ and $\hat p_{11}=n_{11}/(n_{10}+n_{11})$. Under independence the common conditional exceedance probability equals $\hat p=\hat\pi$, and
\begin{equation}\label{eq:LRind}
    \mathrm{LR}_{\mathrm{IND}} = -2\Big[\log\{(1 - \hat p)^{n_{00} + n_{10}}\hat p^{n_{01} + n_{11}}\} - \log\{(1 - \hat p_{01})^{n_{00}}\hat p_{01}^{n_{01}}(1 - \hat p_{11})^{n_{10}}\hat p_{11}^{n_{11}}\}\Big] \sim \chi^2_1.
\end{equation}
Conditional coverage combines the two
\begin{equation}\label{eq:LRcc}
    \mathrm{LR}_{\mathrm{CC}} = \mathrm{LR}_{\mathrm{POF}}+\mathrm{LR}_{\mathrm{IND}} \sim \chi^2_2.
\end{equation}
For interpretability we also report the empirical exceedance rate $\hat\pi$ and the Basel traffic-light classification computed on non-overlapping windows of length $m=250$; for $\tau=0.99$ the green, yellow, and red zones correspond to at most $4$ exceedances, $5$-$9$ exceedances, and at least $10$ exceedances, respectively.

\begin{figure}[ht]
    \centering
    \includegraphics[width=0.9\linewidth]{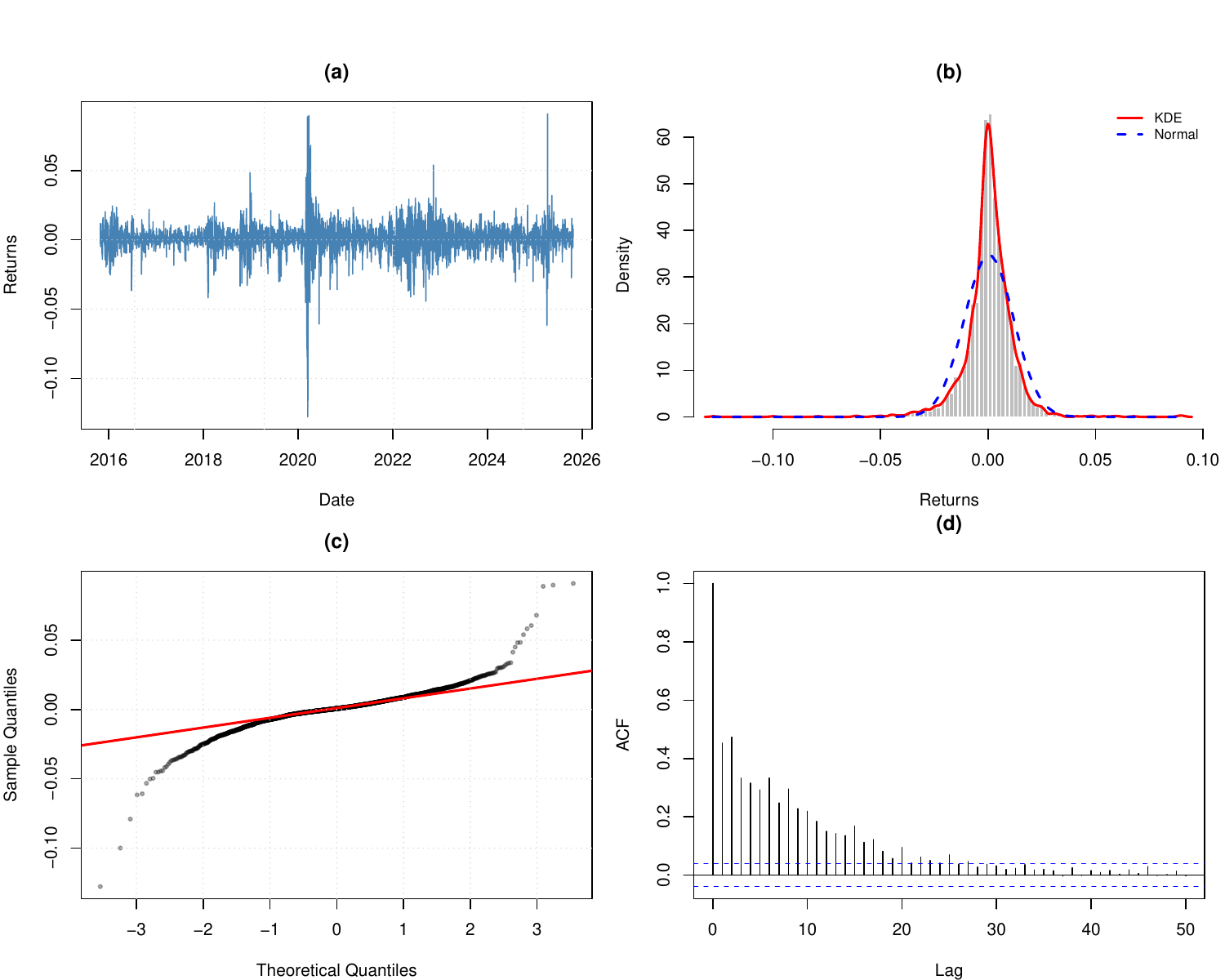}
    \caption{Daily log-returns (a) series; (b) histogram with KDE and normal; (c) normal Q-Q; (d) squared-returns ACF.}
    \label{fig:1}
\end{figure}

\begin{table}[ht]
    \centering
    \caption{Backtesting performance across VaR models}
    \label{tab:backtest_performance}
    \begin{threeparttable}
        \begin{tabular}{lrrrrrrrrr}
            \toprule
             Method           & PoF($p$-value)   & $\pi$(\%) & IND($p$-value) & CC($p$-value)   & green(\%) & yellow(\%) & red(\%) \\
            \midrule
\rowcolor{highlightcolor} FHS & 0.0061(0.9379) & 1.016     & 0(1)         & 0.0064(0.0996) & 89.72     & 10.27      & 0.00 \\
\rowcolor{highlightcolor} HS  & 5.8534(0.0155) & 1.546     & 0(1)         & 5.8645(0.0532) & 52.85     & 39.27      & 8.63 \\
                          MCt & 0.0822(0.7744) & 1.060     & 0(1)         & 0.0834(0.9591) & 79.94     & 20.06      & 0.00 \\
                          QR  & 4.2053(0.4029) & 1.458     & 0(1)         & 4.2146(0.1215) & 57.24     & 37.24      & 8.55 \\
                          EVT & 10.963(0.0929) & 1.767     & 0(1)         & 10.979(0.1302) & 50.19     & 41.16      & 8.63 \\
                        GARCH & 30.073(0.0416) & 2.349     & 0(1)         & 30.100(0.2909) & 30.37     & 64.79      & 4.87 \\
                        EWHS  & 21.710(0.0000) & 0.021     & 0(1)         & 21.733(0.0001) & 40.39     & 51.21      & 8.39 \\
            \bottomrule
        \end{tabular}
            \begin{tablenotes}
            \footnotesize
            \item \textit{Notes:} $\pi=1-\tau$. PoF, IND, CC denote the LR statistics and $p$-values in \eqref{eq:LRpof}, \eqref{eq:LRind}, and \eqref{eq:LRcc}.
            \end{tablenotes}
    \end{threeparttable}
\end{table}

Table \ref{tab:backtest_performance} indicates that FHS combined with SNQESA attains an empirical exceedance rate $\hat\pi=1.016\%$ close to the nominal level, with a very high POF $p$-value ($0.9379$). In contrast, HS exhibits a statistically significant deviation (POF $p$-value $0.0155$, $\hat\pi=1.546\%$). All models show no detectable clustering in the hit process at the daily frequency (IND $p$-values near $1$). Conditional coverage follows the same pattern, with FHS acceptable and HS borderline at the $5\%$ level. The Basel classification places FHS in the green region in the vast majority of rolling years, with no red periods, whereas HS, QR, and EVT have non-negligible red shares.

\begin{figure}[ht]
    \centering
    \includegraphics[width=0.8\linewidth]{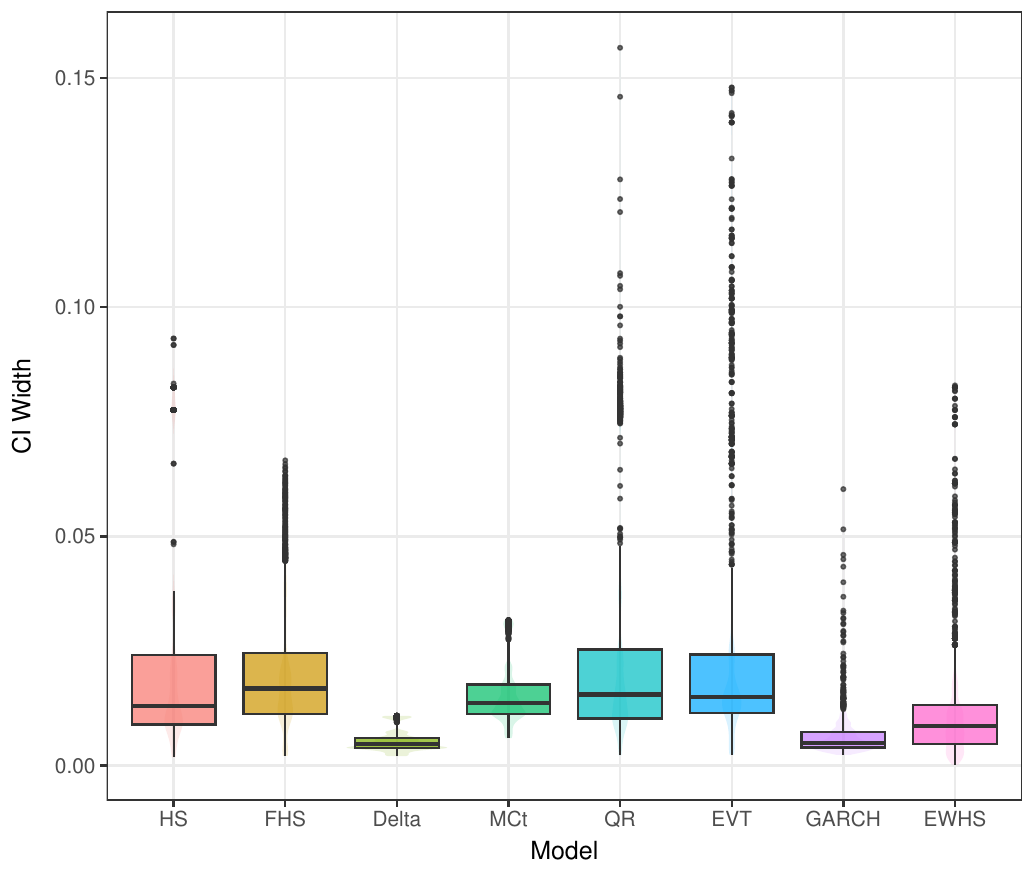}
    \caption{Distribution of VaR interval widths across models.}
    \label{fig:ciwidth}
\end{figure}

Figure \ref{fig:ciwidth} complements Tables \ref{tab:ci_comparison} and \ref{tab:backtest_performance}. Models reporting the narrowest intervals—Delta-Normal and GARCH—also display the largest deviations from nominal coverage (Table \ref{tab:backtest_performance}), indicating over-confidence in the left tail. In contrast, FHS delivers moderate interval widths (Table \ref{tab:ci_comparison}) together with near-nominal unconditional coverage (POF = 0.94), suggesting that its uncertainty quantification is commensurate with forecast risk. MCt and QR/EVT exhibit wider and more dispersed widths, consistent with their heavier-tailed innovations or threshold uncertainty. The EWHS lies between HS and FHS, its slightly wider distribution relative to HS mirrors stronger responsiveness to scale shifts.

\begin{table}[H]
    \centering
    \caption{Confidence-interval characteristics for VaR models}
    \label{tab:ci_comparison}
    \begin{threeparttable}
        \begin{tabular}{lrrrrrrr}
            \toprule
                          Method            & $\bar{W}$ & $se_W$ & $\tilde W$ & $W_{\min}$ & $W_{\max}$ & $\bar t$ (s) \\
            \midrule
\rowcolor{highlightcolor} FHS (SNQESA)      & 0.02023 & 0.01388 & 0.01684 & 0.00214 & 0.06656 & 1.15 \\
\rowcolor{highlightcolor} HS (SNQESA)       & 0.02186 & 0.02190 & 0.01301 & 0.00180 & 0.09314 & 1.12 \\
                          MCt (Bootstrap)   & 0.02468 & 0.02625 & 0.01291 & 0.00551 & 0.03341 & 65.31 \\
                          QR (Bootstrap)    & 0.02405 & 0.02250 & 0.01553 & 0.00226 & 0.14794 & 200.96 \\
                          EVT (Bootstrap)   & 0.02252 & 0.02354 & 0.01426 & 0.00210 & 0.13279 & 86.39 \\
                          GARCH (Bootstrap) & 0.00578 & 0.00393 & 0.00468 & 0.00208 & 0.06081 & 814.3 \\
                          EWHS (Bootstrap)  & 0.01191 & 0.01266 & 0.00864 & 0.00021 & 0.08292 & 95.712 \\
            \bottomrule
        \end{tabular}
        \footnotesize
        \begin{tablenotes}
            \item \textit{Notes:} HS and FHS intervals are constructed by SNQESA; others by bootstrap. $W$(Width) metrics correspond to nominal $95\%$ VaR intervals.
        \end{tablenotes}
    \end{threeparttable}
\end{table}

Table \ref{tab:ci_comparison} compares interval width and computational time. SNQESA produces shorter average intervals for HS and FHS than the bootstrap-based benchmarks (MCt, QR, EVT) at a fraction of the computational cost. Although the GARCH intervals are the narrowest, the unconditional coverage results in Table \ref{tab:backtest_performance} caution against interpreting small absolute widths as genuine efficiency when model uncertainty may be underestimated.

\begin{table}[H]
    \centering
    \caption{Stability metrics for VaR forecast paths}
    \label{tab:stability_assessment}
    \begin{threeparttable}
        \begin{tabular}{lrrrrr}
            \toprule
            Method & VaR Volatility & Change Volatility & Max Drawdown & Turning Ratio & Stability \\
            \midrule
            \rowcolor{highlightcolor}
            FHS   & 0.02681 & 0.00172 & 0.17522 & 0.37593 & 25.43 \\
            \rowcolor{highlightcolor}
            HS    & 0.01868 & 0.00090 & 0.06735 & 0.06191 & 140.95 \\
            MCt   & 0.01537 & 0.00041 & 0.05108 & 0.42857 & 23.10 \\
            QR    & 0.01898 & 0.00162 & 0.09316 & 0.45997 & 21.00 \\
            EVT   & 0.01747 & 0.00070 & 0.06941 & 0.61565 & 16.05 \\
            GARCH & 0.01556 & 0.00466 & 0.21816 & 0.47914 & 19.01 \\
            EWHS  & 0.01685 & 0.00165 & 0.10949 & 0.05039 & 149.46 \\
            \bottomrule
        \end{tabular}
    \begin{tablenotes}
        \footnotesize
        \item \textit{Notes:} VaR Volatility is the standard deviation of $\{\mathrm{VaR}^{(\tau)}_t\}$; Change Volatility is the standard deviation of $\Delta\mathrm{VaR}_t=\mathrm{VaR}^{(\tau)}_{t}-\mathrm{VaR}^{(\tau)}_{t-1}$; Turning Ratio is the fraction of sign changes in $\Delta\mathrm{VaR}_t$; Stability is $1/(\text{Change Volatility}+0.1\times\text{Turning Ratio})$.
    \end{tablenotes}
    \end{threeparttable}
\end{table}

\begin{figure}[ht]
    \centering
    \includegraphics[width=0.85\linewidth]{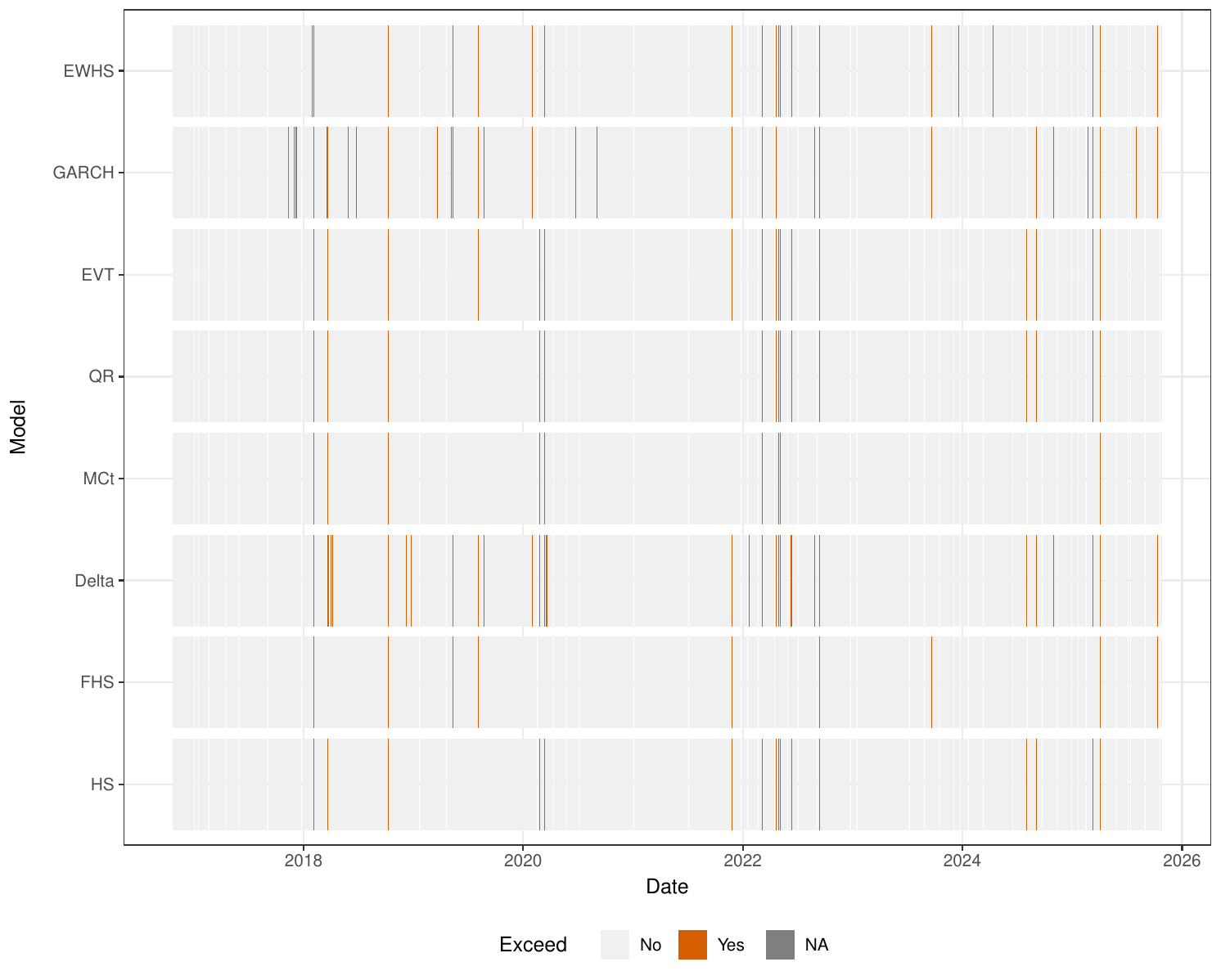}
    \caption{Exceedance heatmap ($I_{t+1} = \mathbf{1}\{\,r_{t+1}\le \mathrm{VaR}^{0.99}_{t+1|t}\,\}$)}
    \label{fig:exceed_heatmap}
\end{figure}

Figure~\ref{fig:exceed_heatmap} visualizes the timing of VaR violations across models. Violations cluster around the 2020 pandemic crash and the 2022 tightening episode, indicating regime shifts rather than idiosyncratic noise. Models differ markedly in how often they breach during these stress windows: Delta-Normal and GARCH display the densest stripes, whereas FHS and MCt are comparatively sparse. This visual ranking matches the unconditional coverage results in Table~\ref{tab:backtest_performance}: FHS attains an empirical exceedance rate close to the nominal 1\% ($\hat \pi=1.016$\%), while Delta, QR, EVT and especially GARCH overshoot the nominal level. EWHS (exponentially-weighted HS) behaves similarly to FHS but reacts more to turning points, which is consistent with its recency weights.

The heatmap also aligns with our interval diagnostics. Models with the densest violation stripes (Delta, GARCH) are those reporting the narrowest VaR intervals in Table~\ref{tab:ci_comparison}, suggesting underestimation of tail uncertainty. Conversely, FHS and MCt combine moderate interval width with fewer breaches. Conditional on a breach, the lollipop plots of exceedance severity show higher tails for Delta/GARCH, reinforcing the caution that narrower intervals do not necessarily imply better efficiency when model risk is material.

Table \ref{tab:stability_assessment} assesses temporal regularity of the forecast sequences. FHS exhibits higher change volatility than HS, indicating greater responsiveness to shifts in conditional scale, while its turning ratio and drawdown do not suggest excessive oscillation. HS displays very low change volatility but, combined with its unconditional coverage deviation, this points to sluggish adjustment to distributional changes.

\begin{table}[htbp]
    \centering
    \small
    \caption{Extreme-event performance (best VaR models, 2018--2023)}
    \label{tab:extreme_events}
    \begin{tabular}{lcccccc}
        \toprule
        Period & Method & Fail & Gap & Ratio & Score & K \\
        \midrule
        2018 Q4 Drawdown        & QR  & 0.1 & $-0.0081$ & 0.983 & 0.925 & 10 \\
        \rowcolor{highlightcolor}
        2018 Volmageddon        & FHS & 0.2 & $-0.0336$ & 0.937 & 0.841 & 10 \\
        \rowcolor{highlightcolor}
        2020 COVID Crash        & FHS & 0.2 & $-0.0557$ & 0.983 & 0.855 & 10 \\
        2020 US Election Vol.   & HS  & 0.0 & $-0.0837$ & 1.000 & 1.000 & 10 \\
        \rowcolor{highlightcolor}
        2022 Fed Tightening Bear& FHS & 0.1 & $-0.0057$ & 0.979 & 0.924 & 10 \\
        2023 US Regional Banks  & HS  & 0.0 & $-0.0333$ & 1.000 & 1.000 & 10 \\
        \bottomrule
    \end{tabular}
    \footnotesize
    \begin{tablenotes}
        \item \textit{Notes:} Fail is the extreme-period exceedance rate; Gap is $\mathbb{E}[\mathrm{Loss}-\mathrm{VaR}\mid \mathrm{Loss}>\mathrm{VaR}]$; Ratio is $\mathbb{E}[\min(\mathrm{VaR}/\mathrm{Loss},1)]$; Score is normalized to $[0,1]$; $K$ is the number of top losses used.
    \end{tablenotes}
\end{table}

Table \ref{tab:extreme_events} focuses on selected stress episodes over 2018-2023. Performance is regime dependent: QR is strongest in the prolonged drawdown of 2018Q4, whereas FHS dominates in volatility-clustering environments and during abrupt left-tail shocks (February 2018, March 2020), as reflected by both failure rates and gap magnitudes. Event-driven short bursts (election week, 2023 regional banking stress) are adequately handled by HS. These contrasts accord with the mechanism of FHS, which couples nonparametric quantiles with a volatility filter and thus adapts to persistent changes in conditional scale while retaining distributional flexibility.

Overall, the evidence supports two points. First, for nonparametric HS-type models, SNQESA provides a statistically accurate and computationally efficient alternative to bootstrap for constructing VaR intervals: it attains comparable or smaller interval widths at a fraction of the cost, with endpoint error of order $n^{-1}$. Second, when combined with a volatility filter (FHS), the approach yields VaR paths with accurate unconditional coverage and stable dynamics across market regimes, including pronounced turbulence, which is desirable for regulatory backtesting and internal risk monitoring.

\section{Discussion}\label{sec:discussion}

SNQESA combines a self-normalized pivot with a constrained empirical saddlepoint approximation, thereby avoiding estimation of $f(q_\tau)$ while aligning tail calibration with the LR/$r^*$ benchmark. In Monte Carlo experiments spanning light-tailed, heavy-tailed and contaminated mixtures, the continuous equal-tailed inversion maintains coverage close to nominal with competitive lengths and no user tuning. The discrete/exact inversion and the minimum-length variant recover nominal coverage in difficult regimes (extreme $\tau$ or coarse rounding) at a predictable increase in width, and the interval-score summaries are consistent with these findings.

Our experience supports the continuous equal-tailed procedure as the default, augmented with a mid-$p$ adjustment and a small ridge in the denominator of the self-normalizer. For extreme quantiles or highly discrete outcomes, \texttt{SNQESA\_disc} or \texttt{SNQESA\_min)} is preferred; the latter is particularly effective when the score distribution is asymmetric. Reporting both endpoints and interval scores helps reveal when improved coverage is achieved primarily through extra width.

The method requires only sorting and solving a low-dimensional tilted system per endpoint. In rolling windows typical for risk management, the runtime is modest and does not rely on resampling loops or bandwidth searches, which makes uncertainty quantification reproducible and easy to audit.

Coupling SNQESA with historical or filtered historical simulation yields VaR bands with accurate unconditional coverage at moderate computational cost. These bands measure estimation uncertainty of the quantile forecast without committing to a parametric volatility model. Narrower bands produced by parametric specifications do not systematically improve POF/IND/CC backtests in our experiments; by contrast, HS/FHS with SNQESA provides stable traffic-light classifications and clearer diagnostics.

Our theory assumes i.i.d.\ sampling and local smoothness of $F$ near $q_\tau$. Coverage may deteriorate if the self-normalizer becomes too small (e.g., near a flat empirical step) or when mass accumulates at $q_\tau$ due to ties/rounding. The ridge safeguards the first case; mid-$p$ and exact inversion address the second. With very small windows and extreme $\tau$, information is intrinsically limited for any method; in such settings we recommend \texttt{SNQESA\_disc} with conservative tail splitting and to report effective exceedance counts.

Extensions we view as feasible. The rank-reduced curvature of the tilted system and the self-normalization suggest that several directions are attainable with limited changes: 
(i) two-sample comparisons and pooled designs by replacing the scalar score with the relevant influence function; 
(ii) regression and panel frameworks (fixed-, random-, and mixed-effects) by plugging in debiased scores and retaining the same constrained tilting step; (iii) weakly dependent data via block/self-normalization of score sums under mixing or martingale-difference assumptions. We expect these extensions to preserve the numerical stability and LR/$r^*$-level calibration that make SNQESA effective in finite samples.

\section*{Disclosure Statement}
The authors report there are no competing interests to declare.

\section*{Funding}
This work was supported by the [Funding Agency] under Grant [number xxxx].

\begin{appendices}
\numberwithin{equation}{section}

\section{Numerical Details}\label{app:num}

This appendix records the full constrained ESA computation, including a closed-form Newton step under rank-1 geometry, stable evaluation of cumulants and curvature, line-search and trust-region safeguards, a vanishing ridge near the boundary, continuity corrections, monotone inversion for confidence limits, as well as complexity and error-propagation calculations. All objects are expressed in the unit-scale CGF $K(\lambda)$ and its derivatives, with the sum-scale recovered by multiplication by $n$ where needed.

For a fixed threshold $t\in\mathbb{R}$ write
\begin{equation*}
    x_{\mathrm{obs}}(t)=\frac{S_n(t)}{\sqrt{Q_n(t)}},\quad 
    S_n(t)=\sum_{i=1}^n\{\tau-\mathbf{1}(X_i\le t)\},\quad 
    Q_n(t)=\sum_{i=1}^n\{\tau-\mathbf{1}(X_i\le t)\}^2,
\end{equation*}
and impose the boundary in the observed direction
\begin{equation*}
    g(s,q)=s-x_{\mathrm{obs}}\sqrt{q}.
\end{equation*}
The two-point unit CGF for $W_i=(\psi_\tau,\psi_\tau^2)^\top$ is
\begin{equation*}
    M(\lambda)=\tau e^{\lambda_1(\tau-1)+\lambda_2(1-\tau)^2}+(1-\tau)e^{\lambda_1\tau+\lambda_2\tau^2},\quad K(\lambda)=\log M(\lambda),
\end{equation*}
with tilted success probability
\begin{equation*}
    p_\lambda=\frac{\tau e^{\lambda_1(\tau-1)+\lambda_2(1-\tau)^2}}{\tau e^{\lambda_1(\tau-1)+\lambda_2(1-\tau)^2}+(1-\tau)e^{\lambda_1\tau+\lambda_2\tau^2}}\in(0,1),
\end{equation*}
and rank-1 derivatives
\begin{equation*}
    \mu(\lambda)=\nabla K(\lambda)=\big(\tau-p_\lambda, \tau^2+p_\lambda(1-2\tau)\big)^\top, \quad
    \Sigma(\lambda)=\nabla^2K(\lambda)=p_\lambda(1-p_\lambda)vv^\top, \quad
    v=(-1,1-2\tau)^\top.
\end{equation*}
On sum-scale $J(\lambda)=n\Sigma(\lambda)=np_\lambda(1-p_\lambda)v v^\top$. The constrained saddlepoint $(\hat\lambda,\hat\eta,\hat\mu)$ solves
\begin{equation}\label{eq:app-F}
    F(\lambda,\eta):=
    \begin{bmatrix}
        g\big(n\mu(\lambda)\big)\\[2pt]
        \lambda-\eta\nabla g\big(n\mu(\lambda)\big)
    \end{bmatrix}
    =
    \begin{bmatrix}
        0\\[2pt]
        \mathbf{0}
    \end{bmatrix},\quad 
    \nabla g(\mu)=\left(1,\ -\frac{x_{\mathrm{obs}}}{2\sqrt{\mu_2}}\right)^\top,\quad
    H_g(\mu)=\begin{pmatrix}0&0\\[2pt]0&\displaystyle\frac{x_{\mathrm{obs}}}{4\mu_2^{3/2}}\end{pmatrix}.
\end{equation}
A damped Newton step $(\delta\lambda,\delta\eta)$ solves $D F(\lambda,\eta)[\delta\lambda,\delta\eta]=-F(\lambda,\eta)$ with Jacobian
\begin{equation*}
    D F(\lambda,\eta)=
    \begin{bmatrix}
        \nabla g(n\mu(\lambda))^\top J(\lambda) & 0\\[2pt]
        I_2 - \eta H_g(n\mu(\lambda))J(\lambda) & -\nabla g(n\mu(\lambda))
    \end{bmatrix}.
\end{equation*}
Because $J(\lambda)$ is rank-1, the Newton system has a closed-form solution by projecting along $v$. Let $\kappa(\lambda)=np_\lambda(1-p_\lambda)$ so that $J(\lambda)=\kappa vv^\top$, and define
\begin{equation*}
    \gamma(\lambda)=\nabla g(n\mu(\lambda))\cdot v,\quad 
    \zeta(\lambda)=v^\top H_g(n\mu(\lambda)) v=\frac{x_{\mathrm{obs}}}{4}\cdot\frac{(1-2\tau)^2}{\big(n\mu_2(\lambda)\big)^{3/2}}.
\end{equation*}
With residuals $r_0=g(n\mu(\lambda))$ and $r_1=\lambda-\eta\nabla g(n\mu(\lambda))$, the first Newton equation gives
\begin{equation}\label{eq:app-alpha}
    v^\top\delta\lambda=-\frac{r_0}{\kappa\gamma},\quad 
    \alpha:=\frac{v^\top\delta\lambda}{\|v\|^2}=-\frac{r_0}{\kappa\gamma\|v\|^2},\quad 
    \delta\lambda=\alpha v+z,\ z\perp v.
\end{equation}
The second equation reduces to
\begin{equation}\label{eq:app-second-row}
    \alpha v+z-\eta\kappa\alpha\|v\|^2 H_g v-\nabla g\delta\eta=-r_1.
\end{equation}
Left-multiplying \eqref{eq:app-second-row} by $\nabla g^\top$ yields
\begin{equation}\label{eq:app-deta}
    -\|\nabla g\|^2\delta\eta+\alpha\gamma+\nabla g\cdot z-\eta\kappa\alpha\|v\|^2\nabla g\cdot(H_g v)=-\nabla g\cdot r_1.
\end{equation}
Setting $z\equiv 0$ gives the explicit update
\begin{equation}\label{eq:app-deta-explicit}
    \delta\eta=\frac{\nabla g\cdot r_1-\alpha\gamma+\eta\kappa\alpha\|v\|^2\nabla g\cdot(H_g v)}{\|\nabla g\|^2},\quad  \delta\lambda=\alpha v.
\end{equation}
If one prefers to remove residual components orthogonal to $\operatorname{span}\{v,\nabla g\}$, a correction $z$ may be computed from \eqref{eq:app-second-row} as
\begin{equation}\label{eq:app-z}
    z=-r_1-\alpha v+\eta\kappa\alpha \|v\|^2 H_g v+\nabla g\delta\eta,\quad 
    z\leftarrow z-\frac{v^\top z}{\|v\|^2}v,
\end{equation}
which enforces $z\perp v$ and preserves stability even when $H_g$ is large ($\mu_2\downarrow 0$).

The raw Newton increment is damped by a backtracking line-search with factor $\rho\in(0,1)$ (we use $\rho=1/2$), decreasing the merit
\begin{equation}\label{eq:app-merit}
    \Psi(\lambda,\eta)=\frac{1}{2}r_0^2+\frac{1}{2}\|r_1\|^2,
\end{equation}
until $\Psi$ decreases (Armijo condition suffices). A maximal damping factor of 10 avoids excessive shortening. To keep the tilt interior we enforce
\begin{equation*}
    p_\lambda\in[\varepsilon,1-\varepsilon],\quad \varepsilon=10^{-10},
\end{equation*}
either by a logistic reparameterization $\beta=\logit(p_\lambda)$ or by rejecting steps that violate the interval; both are equivalent since $\nabla p_\lambda=p_\lambda(1-p_\lambda)v$ is collinear with $v$. We also bound $|\hat\eta|$ via a trust-region on $(\lambda,\eta)$ to prevent overflow when $\nabla g(\hat\mu)\cdot v$ is near zero.

Near extreme quantiles or very small $n$ we stabilize $x_{\mathrm{obs}}$ and $\nabla g$ by a vanishing ridge
\begin{equation}\label{eq:app-ridge}
    Q_n(t)\leftarrow Q_n(t)+\epsilon_n,\quad \epsilon_n\to 0,
\end{equation}
and in practice set $\epsilon_n=cn^{-1/2}$ with a small constant $c$. A first-order Taylor expansion gives
\begin{equation*}
    x_{\mathrm{obs}}^{(\mathrm{ridge})}-x_{\mathrm{obs}}
    =\frac{S_n}{\sqrt{Q_n+\epsilon_n}}-\frac{S_n}{\sqrt{Q_n}}
    =-\frac{S_n\epsilon_n}{2Q_n^{3/2}}+O\!\left(\frac{\epsilon_n^2}{Q_n^{5/2}}\right),
\end{equation*}
and under $H_0$, $S_n=O_p(n^{1/2})$ and $Q_n=n\tau(1-\tau)+O_p(n^{1/2})$, whence
\begin{equation}\label{eq:ridge-size}
    x_{\mathrm{obs}}^{(\mathrm{ridge})}-x_{\mathrm{obs}}=O_p(n^{-3/2}),
\end{equation}
which is negligible relative to the $n^{-1/2}$ scale of tail arguments and does not affect second-order coverage.

Tail areas are computed from the sum-scale signed root deviance
\begin{equation*}
    D_{\text{sum}}=2n\{\hat\lambda^\top\mu(\hat\lambda)-K(\hat\lambda)\},\quad 
    r=\operatorname{sgn}(\hat\eta)\sqrt{D_{\text{sum}}},
\end{equation*}
and curvature terms. Skovgaard’s rank-1 form
\begin{equation*}
    w=|\hat\eta|\frac{\sqrt{\pdet(J(\hat\lambda))}\|\nabla g(\hat\mu)\|}{\big(\nabla g(\hat\mu)^\top\Sigma(\hat\lambda)^+\nabla g(\hat\mu)\big)^{1/2}}
\end{equation*}
reduces here to the closed expression
\begin{equation}\label{eq:w-closed}
    w=|\hat\eta|\sqrt{n}p_{\hat\lambda}(1-p_{\hat\lambda})\frac{\|v\|^3\|\nabla g(\hat\mu)\|}{|\nabla g(\hat\mu)\cdot v|},
\end{equation}
since $\pdet(J)=np_{\hat\lambda}(1-p_{\hat\lambda})\|v\|^2$ and $\Sigma(\hat\lambda)^+=\{p_{\hat\lambda}(1-p_{\hat\lambda})\}^{-1}vv^\top/\|v\|^4$ imply $\nabla g(\hat\mu)^\top\Sigma(\hat\lambda)^+\nabla g(\hat\mu)=\{(\nabla g(\hat\mu)\cdot v)^2\}/\{p_{\hat\lambda}(1-p_{\hat\lambda})\|v\|^4\}$. We also use the signed Lugannani-Rice quantity
\begin{equation}\label{eq:q-signed}
    q^{\pm}=\big(\logit(\hat p)-\logit(\tau)\big)\sqrt{n\hat p(1-\hat p)}.
\end{equation}
The Barndorff-Nielsen correction
\begin{equation*}
    r^*=r+\frac{1}{r}\log\frac{r}{w}
\end{equation*}
is used if $|\log(r/w)|\le c_0$ (default $c_0=2$); otherwise the signed Lugannani-Rice approximation
\begin{equation*}
    p_{\mathrm{dir}}(t)\approx\Phi\big(-\operatorname{sgn}(x_{\mathrm{obs}})r\big)+\phi(r)\left(\frac{1}{r}-\frac{1}{q^{\pm}}\right)
\end{equation*}
is applied to the observed tail. Because under $H_0$ the statistic is a monotone transform of a binomial mean, we evaluate tails by default with a mid-p modification
\begin{equation*}
    p_{\mathrm{mid}}(t)=\mathbb{P}\{T_n(t)>x_{\mathrm{obs}}(t)\}+\frac{1}{2}\mathbb{P}\{T_n(t)=x_{\mathrm{obs}}(t)\},
\end{equation*}
or with a small Cornish-Fisher continuity shift; both preserve the $O(n^{-1})$ coverage error when combined with third-order tail accuracy.

For stable floating-point evaluation write $K(\lambda)=\log\{\tau e^A+(1-\tau)e^B\}$ with $A=\lambda_1(\tau-1)+\lambda_2(1-\tau)^2$, $B=\lambda_1\tau+\lambda_2\tau^2$, and use log-sum-exp
\begin{equation*}
    m=\max\{A,B\},\quad K(\lambda)=m+\log\big(\tau e^{A-m}+(1-\tau)e^{B-m}\big),
\end{equation*}
and compute $\hat\lambda^\top\mu(\hat\lambda)-K(\hat\lambda)$ with log1p and expm1 to reduce cancellation. The Kullback-Leibler divergence and logit difference are evaluated as
\begin{equation*}
    \mathrm{KL}(u\|\tau)=u\log\Big(1+\frac{u-\tau}{\tau}\Big)+(1-u)\log\Big(1-\frac{u-\tau}{1-\tau}\Big),\quad 
    \logit(u)-\logit(\tau)=\log\frac{u(1-\tau)}{\tau(1-u)}.
\end{equation*}
For large positive $z$, compute $\Phi(-z)$ as $\tfrac{1}{2}\mathrm{erfc}(z/\sqrt{2})$; if a log-tail is needed,
\begin{equation*}
    \log\Phi(-z)=-\frac{z^2}{2}-\log z-\frac{1}{2}\log(2\pi)+\log\Big(1-\frac{1}{z^2}+\frac{3}{z^4}-\cdots\Big).
\end{equation*}

The one-dimensional fallback relies on the strictly monotone transform
\begin{equation*}
    h(u)=\frac{\sqrt n(\tau-u)}{\sqrt{u(1-\tau)^2+(1-u)\tau^2}},\quad 
    d(u)=\tau^2+u(1-2\tau),
\end{equation*}
whose derivatives are
\begin{equation*}
    h'(u)=-\frac{\sqrt n}{2d(u)^{3/2}}\{(1-2\tau)u+\tau\}<0,\quad 
    h''(u)=\frac{3\sqrt n(1-2\tau)\{(1-2\tau)u+\tau\}}{4d(u)^{5/2}}.
\end{equation*}
Given $x_{\mathrm{obs}}$, solve $h(u)=x_{\mathrm{obs}}$ for $u_x\in(0,1)$ by safeguarded Newton
\begin{equation*}
    u^{(k+1)}=\Pi_{(a,b)}\Big\{u^{(k)}-\frac{h(u^{(k)})-x_{\mathrm{obs}}}{h'(u^{(k)})}\Big\},\quad
    (a,b)=(10^{-12},1-10^{-12}),
\end{equation*}
with bisection if the Newton step fails the interval-shrinking test, and then evaluate the binomial signed LR quantities at $u_x$:
\begin{equation*}
    r=\operatorname{sgn}(u_x-\tau)\sqrt{2n\mathrm{KL}(u_x\|\tau)},\quad 
    q^{\pm}=(\logit(u_x)-\logit(\tau))\sqrt{n u_x(1-u_x)}.
\end{equation*}

Confidence limits are obtained by inverting the directed tails. Writing $f_\downarrow(t)=p_{\downarrow}(t)-\alpha/2$ and $f_\uparrow(t)=p_{\uparrow}(t)-\alpha/2$, monotonicity of $p_{\downarrow}$ on $(-\infty,\hat q_\tau]$ and of $p_{\uparrow}$ on $[\hat q_\tau,\infty)$ permits bracketing $t_L$ and $t_U$ by scanning ordered $\{X_{(k)}\}$ until a sign change, and refining by bisection to $|t_{k+1}-t_k|<10^{-8}$. With ties or discretization we use the mid-p or a Cornish-Fisher shift; both preserve the $O(n^{-1})$ coverage when combined with third-order tails.

We summarize complexity. Sorting the sample once costs $O(n\log n)$. Along the sorted path, $S_n$ and $Q_n$ admit $O(1)$ updates between adjacent order statistics because $S_n$ changes by $\pm(1-\tau)$ and $Q_n$ by a deterministic affine increment; scanning to bracket each endpoint is $O(n)$ in the worst case and much less in practice. Each $p$-value evaluation solves a $3\times 3$ nonlinear system with rank-1 linear algebra; the number of Newton steps is bounded by a small constant due to damping, hence $O(1)$ per evaluation. Bisection requires $O(\log(1/\mathrm{tol}))$ steps (with $\mathrm{tol}=10^{-8}$ this is at most 27). Therefore, the overall complexity for a single $(1-\alpha)$ CI is $O(n\log n)$ dominated by sorting.

Finally we quantify error propagation. Let $p^\star(t)$ denote the exact directed tail and $\tilde p(t)$ the ESA approximation. For the one-parameter binomial tilt, the signed Lugannani-Rice tail error satisfies
\begin{equation*}
    \sup_{|r|\le C}\big|p^\star - \tilde p_{\mathrm{LR}}\big|=O(n^{-3/2}), \quad 
    \sup_{|r|\le C}\frac{\big|p^\star-\tilde p_{r^*}\big|}{p^\star} = O(n^{-3/2}),
\end{equation*}
uniformly on compact $r$-sets, and by algebraic equivalence the same orders hold for the constrained rank-1 construction. Passing from $p$ to an endpoint $t$ near $q_\tau$ introduces the local scale $f(q_\tau)^{-1}$. With $F(t)=\tau+f(q_\tau)(t-q_\tau)+O((t-q_\tau)^2)$ and
\begin{equation*}
    z(t)=\frac{\sqrt n(F(t)-\tau)}{\sqrt{\tau(1-\tau)}},\quad p^\star(t)=\Phi\big(-z(t)\big)+O(n^{-1}),
\end{equation*}
one has at $t=q_\tau$ the sensitivity
\begin{equation*}
    \frac{dp^\star}{dt}\Big|_{t=q_\tau}=-\phi(0)\frac{\sqrt n f(q_\tau)}{\sqrt{\tau(1-\tau)}}.
\end{equation*}
Therefore a tail error $\Delta p(t)=\tilde p(t)-p^\star(t)=O(n^{-3/2})$ transfers to an endpoint perturbation
\begin{equation*}
    \Delta t=-\frac{\Delta p}{p^{\star\prime}(q_\tau)}+o(\Delta p)=O\Big(\frac{1}{n}\Big)\cdot\frac{\sqrt{2\pi\tau(1-\tau)}}{f(q_\tau)},
\end{equation*}
so the coverage error of the inverted CI is $O(n^{-1})$. The ridge perturbation \eqref{eq:ridge-size} induces a tail change of order $O_p(n^{-3/2})$ and therefore an endpoint impact of order $O_p(n^{-2})$, asymptotically negligible. In extreme-quantile regimes with $\tau=\tau_n$ satisfying $\min\{n\tau_n,n(1-\tau_n)\}\to\infty$, all rates remain valid uniformly over compact $r$-sets; choosing $\epsilon_n=cn^{-1/2}$ keeps the orders unchanged.

Collecting identities used by curvature terms, the pseudo-determinant and generalized inverse are
\begin{equation*}
    \pdet\big(\Sigma(\lambda)\big)=p_\lambda(1-p_\lambda)\|v\|^2,\quad 
    \Sigma(\lambda)^+=\frac{vv^\top}{p_\lambda(1-p_\lambda)\|v\|^4},
\end{equation*}
so any expression originally involving $\det\Sigma$ or $\Sigma^{-1}$ is read along $\operatorname{span}(v)$ by replacing them with $\pdet$ and $\Sigma^+$. With these definitions, the entire tail evaluation at a fixed $t$ follows
\begin{itemize}
    \item compute $x_{\mathrm{obs}}$, solve \eqref{eq:app-F} via \eqref{eq:app-alpha}-\eqref{eq:app-z} with damping and trust-region
    \item evaluate $(r,w,q^{\pm})$ then signed $r^*$ or LR tail
\end{itemize}
with a one-dimensional fallback through $h^{-1}$ and a degeneracy check $|\nabla g(\hat\mu)\cdot v|>\delta$ (we use $\delta=10^{-10}$); if violated we switch to the binomial evaluation at $u_x$ for both $r^*$ and LR tails.

\section{Proofs for Section \ref{sec:theory}}\label{app:proofs-theory}

We collect detailed proofs for Lemma \ref{lem:bin-reduction} and Theorems \ref{thm:tail}-\ref{thm:extreme}. Throughout, recall the definitions in Sections \ref{sec:method} and \ref{sec:theory}. All limits are taken as $n\to\infty$ unless stated otherwise.

\begin{proof}[Proof of Lemma \ref{lem:bin-reduction}]
By definition,
\begin{equation*}
    S_n(t)=\sum_{i=1}^n\{\tau-Y_i(t)\}=n\{\tau-\bar Y_n(t)\},
\end{equation*}
and
\begin{equation*}
    Q_n(t)=\sum_{i=1}^n\{\tau-Y_i(t)\}^2=\sum_{i=1}^n\big[Y_i(t)(1-\tau)^2+\{1-Y_i(t)\}\tau^2\big]=n\big\{\bar Y_n(t)(1-\tau)^2+\{1-\bar Y_n(t)\}\tau^2\big\}.
\end{equation*}
Therefore
\begin{equation}\label{eq:lemma-h-identity}
    T_n(t)=\frac{S_n(t)}{\sqrt{Q_n(t)}}=\frac{\sqrt n\{\tau-\bar Y_n(t)\}}{\sqrt{\bar Y_n(t)(1-\tau)^2+\{1-\bar Y_n(t)\}\tau^2}}=h\big(\bar Y_n(t)\big),
\end{equation}
with $h$ given in \eqref{eq:def-h}. Differentiation yields $h'(u)=-\sqrt n\{(1-2\tau)u+\tau\}/\{2d(u)^{3/2}\}<0$ on $(0,1)$, so $h$ is strictly decreasing and continuous. For any $x\in\mathbb{R}$ the equation $h(u)=x$ has a unique solution $u_x=h^{-1}(x)\in(0,1)$, and since $h$ is decreasing,
\begin{equation*}
    \{T_n(t)\ge x\}=\{h(\bar Y_n(t))\ge x\}=\{\bar Y_n(t)\le u_x\}.
\end{equation*}
This proves the lemma.
\end{proof}

\begin{proof}[Proof of Theorem \ref{thm:tail}]
Fix $t$ and denote $\bar Y_n=\bar Y_n(t)$. Under $H_0:F(t)=\tau$ we have $\bar Y_n=n^{-1}\mathrm{Bin}(n,\tau)$. By Lemma \ref{lem:bin-reduction},
\begin{equation}\label{eq:tail-to-binomial}
    \mathbb{P}\{T_n(t)\ge x\}=\mathbb{P}\{\bar Y_n\le u_x\},\quad u_x=h^{-1}(x).
\end{equation}
Write $S=\sum_{i=1}^n Y_i$, so $S\sim\mathrm{Bin}(n,\tau)$ and $\bar Y_n=S/n$. The CGF of $S$ is
\begin{equation}
    K_S(\theta)=\log\mathbb{E}[e^{\theta S}]=n\log\big\{(1-\tau)+\tau e^\theta\big\}.
\label{eq:KS}
\end{equation}
Let $s=nu_x$. The saddlepoint $\hat\theta$ solves $K_S'(\hat\theta)=s$, i.e.
\begin{equation}
    \frac{d}{d\theta}K_S(\theta)\Big|_{\theta=\hat\theta}=n\cdot\frac{\tau e^{\hat\theta}}{(1-\tau)+\tau e^{\hat\theta}}=s,
\label{eq:theta-hat-eq}
\end{equation}
so the tilted success probability equals $u_x$:
\begin{equation}
\hat p:=\frac{\tau e^{\hat\theta}}{(1-\tau)+\tau e^{\hat\theta}}=u_x.
\label{eq:p-hat-equals-ux}
\end{equation}
The signed one-parameter saddlepoint scalars for the lower tail $\mathbb{P}\{S\le s\}$ are
\begin{equation}
r=\operatorname{sgn}(u_x-\tau)\sqrt{2n\,\mathrm{KL}(u_x\|\tau)},\quad
q^{\pm}=\big(\logit(u_x)-\logit(\tau)\big)\sqrt{n\,u_x(1-u_x)},
\label{eq:r-q-binom-signed}
\end{equation}
and the Lugannani-Rice and Barndorff-Nielsen forms give, uniformly on compact $r$-sets,
\begin{equation}
\mathbb{P}\{S\le s\}=\Phi(r)+\phi(r)\left(\frac{1}{r}-\frac{1}{q^{\pm}}\right)+O(n^{-3/2}),\quad
\mathbb{P}\{S\le s\}=\Phi(r^*)\{1+O(n^{-3/2})\}.
\label{eq:LR-BN-binom}
\end{equation}
By \eqref{eq:tail-to-binomial}-\eqref{eq:LR-BN-binom}, for $x=x_{\mathrm{obs}}(t)$ with sign directing the observed tail,
\begin{equation}
p_{\mathrm{dir}}(t)=
\begin{cases}
\Phi(r)+\phi(r)\left(\dfrac{1}{r}-\dfrac{1}{q^{\pm}}\right)+O(n^{-3/2}),& x\ge 0,\\[8pt]
\Phi(-r)+\phi(r)\left(\dfrac{1}{r}-\dfrac{1}{q^{\pm}}\right)+O(n^{-3/2}),& x<0,
\end{cases}
\label{eq:dir-tail-binom}
\end{equation}
which can be summarized as $p_{\mathrm{dir}}(t)=\Phi(-\operatorname{sgn}(x)r)+\phi(r)\big(1/r-1/q^{\pm}\big)+O(n^{-3/2})$. For the constrained rank-1 ESA we have $\hat\mu=(\tau-\hat p,\ \tau^2+\hat p(1-2\tau))^\top$ and the boundary $g(\hat\mu)=0$ reads
\begin{equation}
\tau-\hat p=x_{\mathrm{obs}}\sqrt{\tau^2+\hat p(1-2\tau)}.
\label{eq:h-equality}
\end{equation}
Hence $h(\hat p)=x_{\mathrm{obs}}$ and, by strict monotonicity of $h$, $\hat p=u_x$. Therefore the constrained deviance, $r$, $q^{\pm}$ and $r^*$ match the binomial ones, which proves \eqref{eq:rstar-accuracy}-\eqref{eq:lr-accuracy}.
\end{proof}

\begin{proof}[Proof of Theorem \ref{thm:coverage}]
Let $p^\star(t)$ denote the exact directed tail at $t$, and $\tilde p(t)$ the ESA approximation from Theorem \ref{thm:tail}. Near $t=q_\tau$,
\begin{equation}\label{eq:local-tail-expansion}
    p^\star(t)=\Phi\big(-z(t)\big)+O(n^{-1}),\quad z(t)=\frac{\sqrt n\{F(t)-\tau\}}{\sqrt{\tau(1-\tau)}},
\end{equation}
uniformly for $t$ in a shrinking neighborhood of $q_\tau$. Differentiating at $t=q_\tau$ gives
\begin{equation} \label{eq:dpdt-again}
    \frac{dp^\star}{dt}\Big|_{t=q_\tau}=-\phi(0)\,\frac{\sqrt n\,f(q_\tau)}{\sqrt{\tau(1-\tau)}}.
\end{equation}
By Theorem \ref{thm:tail}, $\tilde p(t)=p^\star(t)+\Delta_n(t)$ with
\begin{equation}\label{eq:tail-error}
    \sup_{|r|\le C}|\Delta_n(t)|=O(n^{-3/2}),
\end{equation}
uniformly in $t$ such that the directed $r$ lies in a compact set. Consider the lower endpoint $t_L$ defined by $p_{\downarrow}(t_L)=\alpha/2$; write $t_L^\star$ for the exact solution and $\tilde t_L$ for the ESA solution. A first-order expansion of $p^\star$ around $t_L^\star$ together with \eqref{eq:dpdt-again} and \eqref{eq:tail-error} yields
\begin{equation}  \label{eq:endpoint-error}
    \tilde t_L-t_L^\star=\frac{\Delta_n(\tilde t_L)}{p^{\star\prime}(t_L^\star)}+o(\Delta_n)=O\Big(\frac{1}{n}\Big)\cdot\frac{\sqrt{\tau(1-\tau)}}{f(q_\tau)\phi(0)}=O(n^{-1}),
\end{equation}
and similarly for the upper endpoint $t_U$. Therefore the ESA interval differs from the exact equal-tailed interval by $O(n^{-1})$ at each endpoint, and the coverage error satisfies the same order; the uniformity over classes with $f(q_\tau)$ bounded and bounded away from zero follows from the uniformity in \eqref{eq:tail-error} and the monotonicity of the directed tails.
\end{proof}

\begin{proof}[Proof of Theorem \ref{thm:heavy}]
Under $H_0$, the bivariate variable $W_i=(\psi_\tau,\psi_\tau^2)$ has a two-point distribution supported on $w_0=(\tau,\tau^2)^\top$ and $w_1=(\tau-1,(1-\tau)^2)^\top$. Hence its MGF and CGF exist for all arguments regardless of the tail behavior of $X$, and the rank-1 derivatives are always well defined. The third-order tail accuracy \eqref{eq:rstar-accuracy}-\eqref{eq:lr-accuracy} and the coverage order \eqref{eq:coverage-order} only require the binomial reduction of Lemma \ref{lem:bin-reduction} and the local continuity $f(q_\tau)>0$ to linearize $F(t)$ near $q_\tau$. No global moment assumptions on $X$ are used. Therefore the stated robustness holds.
\end{proof}

\begin{proof}[Proof of Theorem \ref{thm:extreme}]
Let $\tau=\tau_n$ with $\min\{n\tau_n,n(1-\tau_n)\}\to\infty$. Then the effective information $n\tau_n(1-\tau_n)$ diverges, and the one-parameter binomial saddlepoint expansions \eqref{eq:LR-BN-binom} hold uniformly on compact $r$-sets. By Lemma \ref{lem:bin-reduction} and the identity $\hat p=u_x$, the constrained rank-1 ESA inherits the same scalars $r$, $q^{\pm}$ and $r^*$ uniformly, establishing \eqref{eq:rstar-accuracy}-\eqref{eq:lr-accuracy} for $\tau_n$.

For coverage, write $z(t)$ as in \eqref{eq:local-tail-expansion} with $\tau$ replaced by $\tau_n$, so that $z'(q_{\tau_n})$ is of order $\sqrt{n\tau_n(1-\tau_n)}\,f(q_{\tau_n})$. The same linearization argument as in \eqref{eq:endpoint-error} shows that endpoint errors remain $O(n^{-1})$ uniformly as long as $f(q_{\tau_n})$ stays bounded and bounded away from zero on a neighborhood of $q_{\tau_n}$. If a vanishing ridge $Q_n\leftarrow Q_n+\epsilon_n$ is used with $\epsilon_n\to 0$ (e.g. $\epsilon_n=c\,n^{-1/2}$), the perturbation of $x_{\mathrm{obs}}$ is $O_p(n^{-3/2})$ as in \eqref{eq:ridge-size} and therefore negligible for both third-order tails and second-order endpoints. This proves the theorem.
\end{proof}

\end{appendices}

\bibliography{ref}
\end{document}